\documentclass[preprintnumbers,amsmath,amssymb,aps,jmp,nofootinbib,longbibliography,12pt]{revtex4-2}%{revtex4-1}

\usepackage[utf8]{inputenc}
\usepackage[colorlinks,hyperindex]{hyperref}  % Si on compile en PdfLaTeX

\hypersetup
 {
   colorlinks,%
   citecolor=blue,%
   linkcolor=blue,%
   urlcolor=blue,%
 }

\def\setR{\mathbb{R}} 
\DeclareMathOperator{\sgn}{sgn}
\newcommand{\sss}[1]{\scriptscriptstyle #1}

\usepackage{amsthm}
    \theoremstyle{plain}
        \newtheorem{theorem}{Theorem}
        \newtheorem{proposition}{Proposition}

    \theoremstyle{definition}
        \newtheorem{definition}{Definition}
        \newtheorem*{notation}{Notation}

    \theoremstyle{remark}
        \newtheorem*{remark}{Remark}

\begin{document}
%\preprint{APS/123-QED}

\title{Restriction of Laplace operator
on one-forms: from $\setR^{n+2}$ and $\setR^{n+1}$ ambient spaces to embedded (A)dS$_n$  submanifolds}

\author{E.~Huguet$^1$}
\author{J.~Queva$^2$}
\author{J.~Renaud$^3$}
\affiliation{$1$ - Universit\'e de Paris, APC-Astroparticule et Cosmologie (UMR-CNRS 7164), 
Batiment Condorcet, 10 rue Alice Domon et L\'eonie Duquet, F-75205 Paris Cedex 13, France.\\
$2$ - Universit\'e de Corse -- CNRS UMR 6134 SPE, Campus Grimaldi BP 52, 20250 Corte, France.\\
$3$ - Universit\'e Gustave Eiffel, APC-Astroparticule et Cosmologie (UMR-CNRS 7164), 
Batiment Condorcet, 10 rue Alice Domon et L\'eonie Duquet, F-75205 Paris Cedex 13, France.
} 
\email{huguet@apc.univ-paris7.fr, \\queva@univ-corse.fr,\\jacques.renaud@apc.univ-paris7.fr}

\date{\today}% It is always \today, today,


\begin{abstract}
The Laplace-de~Rham  operator acting on a one-form $a$: 
$\square a$, in $\setR^{n+2}$ or  $\setR^{n+1}$ spaces is restricted to $n$-dimensional 
pseudo-spheres. This includes, in particular, the $n$-dimensional de~Sitter and Anti-de~Sitter space-times. 
The restriction is designed to extract the corresponding $n$-dimensional Laplace-de~Rham operator acting on the corresponding $n$-dimensional one-form on pseudo-spheres. Explicit formulas relating these two operators are given 
in each situation. The converse problem, of extending an $n$-dimensional operator composed of the sum of the Laplace-de~Rham operator and additional terms to ambient spaces Laplace-de~Rham operator, is also studied. 
We show that for any additional term this operator on the embedded space is the restriction of Laplace-de~Rham operator on the embedding space.
These results are translated to the Laplace-Beltrami operator thanks to the Weitzenböck formula, for which a proof is also given.

\end{abstract}

\maketitle

%\tableofcontents
\section{Introduction}\label{Sec-Introduction}

This paper deals with the restriction from different embedding spaces ($\setR^{n+2}$ and $\setR^{n+1}$) to the $n$-dimensional 
de~Sitter (dS) and Anti-de~Sitter (AdS) 
space-times  of the Laplace-de~Rham operator, hereafter generically denoted 
by $\square := -(d\delta + \delta d)$, acting either on one-form or on scalar fields.

Embedding curved space-times in flat higher dimensional space is an old and still pursued topic (see \cite{Pavsic:2000qy} for a wealth of references and \cite{Paston:2013uia,Akbar:2017vja,Dunajski:2018xoa,Hong:2020dow} for current endeavors), often with an interest in the description of geodesics and determination of Hawking temperature.
Here, on the contrary, the present work is focused on the restriction of differential operators from the embedding space to the space-time (the embedded submanifold).
Such an approach has been used in the study of quantum field theories from an ambient space point of view \cite{Dirac:1935zz, Dirac:1936fq, Fronsdal:1965zzb, Fronsdal:1974ew, Fronsdal:1975eq, PhysRevD.12.3819, Fronsdal:1978vb, Fang:1979hq, Gazeau:1987nu,Gazeau:1999xn,Garidi:2003bg,Huguet:2006fe,queva:tel-00503186,Faci:2009un, Huguet:2013tv,Takook:2014paa,Zapata:2017gqg,Pethybridge:2021rwf}. 
It has also been already used in the case of quantum mechanics constrained to a surface in which the restricted (scalar) laplacian shows additional terms due to the embedding \cite{JENSEN1971586, Costa_1986,10.1143/ptp/88.2.229, PhysRevA.47.686}.

The lowering of dimensionality,  between embedding and embedded spaces, leads to constraints or additional degrees of freedom for fields 
on which operators act, depending on the situation. For instance, the restriction of $\square A$, $A$ being a vector field, to a lower dimensional embedded manifold, implies that the number of components of $A$ has to be reduced. 
For the restricted operator, these ``additional" components become ``additional" degrees of freedom, they may be used to 
impose conditions such as, 
for instance, transversality or homogeneity on the field. 
This kind of method is of common use when the restricted operator is already known.  See, in particular, the work initiated by Dirac \cite{Dirac:1935zz, Dirac:1936fq}, extended by Fronsdal \cite{Fronsdal:1965zzb, Fronsdal:1974ew, Fronsdal:1975eq, PhysRevD.12.3819, Fronsdal:1978vb, Fang:1979hq} and continued by others \cite{Gazeau:1987nu,Gazeau:1999xn,Garidi:2003bg,Huguet:2006fe,queva:tel-00503186,Faci:2009un, Huguet:2013tv,Takook:2014paa,Zapata:2017gqg,Pethybridge:2021rwf},
where one builds in the embedding space a tensor calculus isomorphic to that of the embedded space. 

Here, we take a different approach,  no conditions are imposed on the one-form or scalar fields  in the embedding space (neither transversality nor homogeneity), and the restricted operator retains its most general form.
To summarize, an explicit relation between the Laplace-de~Rham operator on the embedding space and that on the embedded
space is proved for both $n+1$ and $n+2$ dimensional cases.  Precisely, the restricted operator is decomposed into two parts, the former being the  
Laplace-de~Rham operator of the considered pseudo-sphere, the latter contains terms involving intrinsic geometrical structures (normal 
derivative to the manifold, Lie derivative along dilation vectors,\ldots) and the field over the embedding space (that is unrestricted 
to the manifold). This last term, that we shall denote by AT 
(for ``Additional Term") in the sequel, is thus determined through 
\begin{equation*}
(\square_d \alpha)_{\sss \Sigma}=\square_{\sss \Sigma}\alpha_{\sss \Sigma} + \mbox{AT},
\end{equation*}
where $\square_d$ is the Laplace-de~Rham operator on the $d$-dimensional embedding space, $\Sigma$ is the embedded space, $\square_{\sss \Sigma}$ is the Laplace-de~Rham operator on $\Sigma$, and $(\square_d \alpha)_{\sss \Sigma}$,
  $\alpha_{\sss \Sigma}$, are restrictions to the embedded submanifold, in a sense to be  precised later, of respectively $\square_d \alpha$
 and $\alpha$.

The generic form of AT encompass many known situations as, in particular, the one 
where the restricted operator
corresponds to a one-form (or scalar) field which belongs to an unitary irreducible representations of the 
(A)dS group \cite{Huguet:2016szt}. Moreover, this form allows us to draw general conclusions 
on the freedom left in the interplay between operators in embedding and embedded spaces. 
Precisely, the geometry of the embedding being given, the additional term AT appearing in the  
restriction from the embedding space to the embedded  pseudo-sphere is completely determined by  the field. 
Conversely, if one considers on the embedded pseudo-sphere an operator $\square_{\sss \Sigma}\beta+\chi(\beta)$ where $\beta$  is a scalar field or a 1-form on $\Sigma$, and $\chi$ is any expression depending (or not) on $\beta$  
we will show that there always exists an extension $\varpi(\beta)$ of the field $\beta$ to the embedding space for which  AT is $\chi$, explicitly:
\begin{equation*}
(\square_d \varpi(\beta))_{\sss \Sigma}=\square_{\sss \Sigma}\beta + \chi(\beta).
\end{equation*}
There is therefore no constraint on the additional term which can be any desired expression.

The considered embedding spaces $\setR^{n+2}$ (respectively $\setR^{n+1}$)  are the  $n+2$ (respectively $n+1$) dimensional real vector spaces endowed with the metric 
$\eta$ left invariant under the conformal group SO$(2,n)$ (respectively the isometrical groups SO$(1,n)$ or SO$(2,n-1)$ depending on whether we consider dS or AdS space-time).  Note that, we consider the embedding space $\setR^{n+2}$ as our main goal to allow us further
generalizations (see the conclusion  Sec.~\ref{SEC-Conclusion}). In this respect, the embedding space $\setR^{n+1}$  
appears  as an intermediate step.

The paper is organized as follows.
In Sec.~\ref{Sec-GeomAndMethod} we summarize the steps and the methods 
used to derive the expression of the restricted operators, Sec.~\ref{Sec-GeneralExpressionofBoxa} gives
the expression of the Laplace-de~Rahm operators on one-forms and scalars in terms of general frames and
anholonomy coefficients. The restriction from $n+1$-dimensional space is considered 
in Sec.~\ref{SEC-Rn+1}, it is generalized to the restriction from  $n+2$-dimensional space in Sec.~\ref{SEC-Rn+2}.
The relation between restriction and extension is examined in Sec.~\ref{SEC-AdditionnalTerms}. In Sec.~\ref{SEC-Conclusion} we conclude and consider the possibility of the extension of this work to other space-times.
In appendix \ref{APP-FUE} a summary of differential geometry and its most used formulas, relevant to this article, is given.
The Weitzenböck formula is provided, with a proof in our notations, enabling us to extend our results from the Laplace-de Rham operator to the Laplace-Beltrami operator.
The calculation of the co-differential of a co-frame $\delta(e^a)$,
used in particular in Sec.~\ref{Sec-GeneralExpressionofBoxa}, is detailed in appendix \ref{APP-CalcDelta(e^a)}. 
The calculations of the formulas stated in the paper are detailed in \ref{APP-CalcEqBoxPhiRepOrth}--\ref{APP-IntrinsicFormOfEqScal}.

Throughout the paper except otherwise specified 
$a,b,c$ denote abstract indexes, while $\mu, \nu, \ldots = 0, \ldots, n-1$,  
and $A, B, \ldots = 0, \ldots, n$ and, finally, $\alpha, \beta, \ldots = 0, \ldots, n+1$ are related, respectively, to $\Sigma_n$, $\setR^{n+1}$, and $\setR^{n+2}$.
The canonical coordinates of a point $y$ of $\setR^{n+2}$ (respectively $\setR^{n+1}$) are denoted $\{y^{ \alpha}\}$ (respectively $\{y^{\sss A}\}$)  
and their associated Cartesian ortho-normal basis is denoted $\{\partial_\alpha\}$ (respectively $\{\partial_{\sss A}\}$). 

\section{Geometric context and method}\label{Sec-GeomAndMethod}
We proceed in three main steps. We first express the Laplace-de~Rham operator on a $d$-dimensional pseudo-Riemannian oriented manifold $(M,g)$, in terms of any ortho-normalized (local) field of frames and their associated anholonomy coefficients.   

The  $n+1$-dimensional case is then tackled. The  metric manifold $(M,g)$ is here the  embedding space $\setR^{n+1}$ endowed with
a pseudo-metric $\eta$. We define specific frames of $\setR^{n+1}$ adapted to 
any submanifold of $M$. They are built from orthonormal frames of $M$ completed by its outgoing normal, 
and extended to the embedding space. Anholonomic coefficients of theses frames are then computed 
 for the pseudo-sphere. 
The expressions of, both scalar and one-form,  Laplace-de~Rham operator obtained before are then particularized to  these frames 
and coefficients and then restricted to  the $n$-dimensional pseudo-sphere. The resulting expressions 
are rewritten using intrinsic, frame independent, quantities.

The restriction of the Laplace-de~Rham operator to $n$-dimensional (A)dS space from the 
embedding $\setR^{n+2}$ space is finally considered. The (A)dS space is obtained as the intersection of the null cone of $\setR^{n+2}$ and an  $n+1$-dimensional hyper-plane. In that hyper-plane,  the  (A)dS space is a pseudo-sphere and
we recover the geometry of the $\setR^{n+1}$ case. The frames designed for the restriction in the $\setR^{n+1}$ situation are 
extended to frames adapted to the $\setR^{n+2}$ situation. The restriction of operators is performed in complete analogy to the $\setR^{n+1}$ case. Here again, the restricted operators are re-expressed under an intrinsic form.

\section{Expression of $\square a$ and $\square \phi$ in ortho-normal frames}\label{Sec-GeneralExpressionofBoxa}

In order to restrict the Laplace-de Rham operator $\square$ we first need an expression of it which we can work with.
This is achieved by expressing it in an orthonormal basis with its anholonomy coefficients.

\begin{proposition}
Let $(M, g)$  a $d$-dimensional pseudo-Riemannian manifold endowed (at least locally) by a (pseudo-) 
ortho-normalized frame $e_a,\ a=1,\ldots,d$, with structure coefficients $c^a_{bc}$ which satisfy
\begin{equation*}
    [e_a,e_b]=c^c_{ab}e_c \mbox{ and }de^a=-\frac{1}{2}c^a_{bc}e^b\wedge e^c.
\end{equation*}
Then,  on a one-form or a scalar field, $a$ and $\phi$ respectively,
the Laplace-de~Rham operator reads
\begin{align}
\square a= [&
\eta^{bc} e_c(e_b(a_a)) %% 1
+ c^c_{a b}e_c(a^b) %% 2
%+2c^p_{pb}e_a(a^b) %% 3
+ \eta^{bc}c^d_{a c}e_b(a_{d}) %% 4
- \eta^{bd} \eta_{na}c^n_{cd}e_b(a^c) \nonumber \\ %% 5
&  
+ \eta^{bc}   c^p_{pc}e_b(a_a) %% 6
-\eta^{bd}a_ce_b(c^c_{da}) %% 7
- \eta^{bm}a_c c^c_{ba}c^p_{pm} %% 8
+a^be_a(c^p_{pb})       \nonumber\\ %% 9
&
- \frac{1}{2}a_c\eta^{mf}\eta^{bd}\eta_{na}c^c_{mb}c^n_{fd}]e^a, \label{EQ-FormGeneralCoeffStruct}%% 10 
\end{align}
and
\begin{equation}\label{EQ-ScalarGeneralCoeffStruct}
 \square\phi =\eta^{ab}[e_ae_b(\phi)+e_a(\phi) c^p_{pb}].%,   
\end{equation}
\end{proposition}
\begin{proof}
   Eq.~\eqref{EQ-ScalarGeneralCoeffStruct} is obtained through a straightforward calculation detailed in appendix \ref{APP-CalcEqBoxPhiRepOrth}.
   A lengthier calculation detailed in appendix \ref{APP-CalcEqBoxaRepOrth}, making use of  Eq.~\eqref{EQ-ScalarGeneralCoeffStruct},  establishes Eq.~\eqref{EQ-FormGeneralCoeffStruct}.
\end{proof}

\section{Restriction from $\setR^{n+1}$}\label{SEC-Rn+1}
\subsection{Adapted frames}\label{SUBSEC-AdaptedFramesRn+1}
Let $M_n$ be any $n$-dimensional submanifold  isometrically embedded in some oriented $d$-dimensional pseudo-euclidian space $E$. 
In this situation, for our purpose, it will be convenient to use a so-called local \emph{adapted frame}, that is to say a field of 
frames of $E$ adapted to this embedding. This is defined as follows.
\begin{definition}
Let $Q$ be a point of $M_n$ and $U_{\sss Q}$ a $E$-neighborhood of $Q$. A field of positively oriented ortho-normal frames of $E$:
$e_{\sss A}$, $A=0,\ldots,d-1$ on $U_{\sss Q}$ is said to be \emph{adapted to} $M_n$ iff
at each point of $V_{\sss Q}:=M_n \cap U_{\sss Q}$,   $e_\mu, \mu= 0,\ldots, n-1$ is a local field of  ortho-normal frames of $M_n$.
\end{definition}
\begin{remark}
     From this definition, necessarily $e_n,\ldots,e_{d-1}$ are orthogonal to $M_n$.
\end{remark}

In this section, we restrict $M_n$ to be
the $n$-dimensional pseudo-sphere denoted $\Sigma_n$ in $\setR^{n+1}$  and defined through 
\begin{equation*} 
y^2 := y^{\sss A}y_{\sss A} = -\epsilon H^{-2},
\end{equation*}
$H$ being a positive constant and
$\epsilon = \pm 1$. We assume that $\setR^{n+1}$ is oriented and that $\Sigma_n$ is oriented thanks to its outer normal.
In the following, we focus on the de~Sitter space obtained with $\epsilon=1$, the signature of $E$ being $(1,n)$, and the 
Anti-de~Sitter space  obtained with $\epsilon=-1$, the signature of $E$ being $(2,n-1)$.
Nevertheless, the final results of this section, Eq.~(\ref{EQ-RestricFormRn+1IntrinsicVersion}) and Eq.~(\ref{EQ-RestricScalarRn+1IntrinsicVersion}),  concerning the 
$\setR^{n+1}$ embedding remain valid in other cases and, in particular, for the euclidean sphere $S_n$  obtained with 
$\epsilon=-1$, the signature of $E$ being $(n+1,0)$. This is no more the case concerning the $\setR^{n+2}$ embedding.

\noindent\textbf{Adapted frame for $\Sigma_n$ in $\setR^{n+1}$.}
A set of adapted frames of $\Sigma_n$   can be constructed on $\setR^{n+1}$ as follows. In a neighborhood (in the sense of 
$\Sigma_n$)
$V_{\sss Q}$  of some point $Q$ of $\Sigma_n$, let $e_\mu$ be
a field of direct ortho-normal frames of $\Sigma_n$  defined at each point of $V_{\sss Q}$. This field of frames is then extended 
to a field on $\Sigma_n$ of frames of $\setR^{n+1}$ by 
adding the unit outer normal to $\Sigma_n$: $e_n $ which is, from the very definition of $M_n$, non isotropic.
Then, this whole field of frames ${e_{\sss A}}$ is extended to each point of the cone $U_{\sss Q}=\mathcal{C}(V_{\sss Q})$ of 
$\setR^{n+1}$, which is the union  of all the open half-lines coming from the origin of $\setR^{n+1}$ and crossing $V_{\sss Q}$. 
This is done by moving %\sout{translating}
each frame identically to itself along the half-lines that intercept the origin of the frame. To this end, each component, in the 
canonical basis of $\setR^{n+1}$, of a vector that belongs 
to the frame of $U_{\sss Q}$  is required to be homogeneous of degree zero. 
Namely,  the $e_{\sss A}^{\sss B}$, defined through $e_{\sss A}=e_{\sss A}^{\sss B}\partial_{\sss B}$, are homogeneous of degree 
zero and, more explicitely, since for $y\in U_{\sss Q}$, we have $y/\sqrt{|y|^2}\in V_{\sss Q}$ we \emph{define} $e_{\sss A}^{\sss 
B}(y)$ through $e_{\sss A}^{\sss B}(y)=e_{\sss A}^{\sss B}(y/\sqrt{|y|^2})$.
We note that the unit outer normal explicitly reads
\begin{equation*}
    e_n = \frac{1}{\sqrt{|y^2|}}D,
\end{equation*}
where $D$ is the dilation vector $D = y^{\sss A}\partial_{\sss A}$.
We finally remark, that these choices lead to \linebreak $\eta_{nn} = \eta^{nn} = \sgn(y^2) = - \epsilon $ (in a neighborhood of
$Q$).
The adapted frame just built will be used in the following sections.

\subsection{Coefficients of anholonomy}\label{SUBSEC-CoeffAnholoRn+1}
We now compute the coefficients of anholonomy in the  field of adapted frames defined Sec.~\ref{SUBSEC-AdaptedFramesRn+1} 
using the general relation $[e_b,e_c]= c^a_{bc} e_a$. 

We first note that since $\{e_\mu\}$ is a frame of $\Sigma_n$,  a sub-manifold of $\setR^{n+1}$,
one has $[e_\mu,e_\nu]= c^\lambda_{\mu\nu} e_\lambda$. A first consequence is that the coefficients with Greek-index only 
do not change when restricted to $\Sigma_n$, they will appear inside the Laplace-de~Rham operator on $\Sigma_n$ only and there is 
no 
need to compute them explicitly, a second consequence is that $c^n_{\mu\nu}=0$.
Then the anholonomy coefficients we have to calculate are $c^\nu_{\mu n}$ and $c^n_{\mu n}$.

\begin{proposition}
    \begin{equation*}
         c^\nu_{\mu n}=\frac{1}{\sqrt{|y^2|}}\delta^\nu_\mu \mbox{ and } c^n_{\sss AB}=0.
    \end{equation*}
\end{proposition}
\begin{proof}
To obtain $c^\nu_{\mu n}$ we recall from Sec.~\ref{SUBSEC-AdaptedFramesRn+1} that the coefficients $e_\mu^\nu$ are homogeneous functions 
of degree zero, then  on an arbitrary homogeneous function $\varphi$ of degree $r$ one has:
\begin{align*}
[e_n,e_\mu]\varphi&=\frac{1}{\sqrt{|y^2|}}De_\mu\varphi-e_\mu\frac{1}{\sqrt{|y^2|}}D\varphi\\
&=\frac{r-1}{\sqrt{|y^2|}}e_\mu\varphi-e_\mu\left(\frac{1}{\sqrt{|y^2|}}\right)r\varphi -\frac{1}{\sqrt{|y^2|}}re_\mu\varphi\\
&=- \frac{e_\mu}{\sqrt{|y^2|}}\varphi,
\end{align*}
where, in the r.h.s of the second line, we used the fact that $e_\mu(y^2)=0$ since $e_\mu(y^2) = \langle dy^2,e_\mu\rangle$.  The anholonomy coefficients 
do not depend on the choice of $\varphi$, then the last line of the above calculation shows that
\begin{equation*}
     c^\nu_{\mu n}=\frac{1}{\sqrt{|y^2|}}\delta^\nu_\mu \mbox{ and } c^n_{\sss AB}=0,
\end{equation*}
as stated.
\end{proof}
\begin{remark}
    This expression leads to $c^{\sss A}_{{\sss A} n}= n/\sqrt{|y^2|}$ which corresponds to the result of Eq.~(\ref{EQ-Delta(e^a)}).
\end{remark}

Finally, the derivatives of the coefficients of anholonomy are obtained thanks to their homogeneity: since $e_{\sss A}^{\sss B}$ is homogeneous of degree
zero, $e_{\sss A} = e_{\sss A}^{\sss B} \partial_{\sss B}$, $\{\partial_{\sss B}\}$ being the canonical basis of $\setR^{n+1}$, is homogeneous of degree $-1$, it
follows from the general relation $[e_b,e_c]= c^a_{bc} e_a$, \linebreak that $c^{\sss A}_{\sss BC}$ are homogeneous of degree $-1$, consequently
\begin{equation*}
     e_n(c^{\sss A}_{\sss BC})=-\frac{1}{\sqrt{|y^2|}}c^{\sss A}_{\sss BC}\ \mbox{ and }\ e_\lambda(c^\nu_{\mu n})=0.
\end{equation*}

\subsection{Restriction to the pseudo-sphere $\Sigma_n$}
\subsubsection{Restriction of the one-form operator}\label{SUBSEC-RestRn+1OfOneForm}
We now restrict the one-form $\square_{n+1}a$ to the sphere $\Sigma_n$. This is obtained thanks to the pullback $l^*$ of the 
canonical injection $l: \Sigma_n \rightarrow \setR^{n+1}$.
\begin{notation}
If $\alpha$ is a $p$-form on $\setR^{n+1}$ we set
$\alpha_{\sss\Sigma} := l^*\alpha$. 
\end{notation}
From the definition of an adapted frame Sec.~\ref{SUBSEC-AdaptedFramesRn+1} one has $l^* e^\mu 
= e^\mu$ and 
$l^* e^n = 0$, then for $\alpha$ a one-form:
\begin{equation*}
    \alpha_{\sss \Sigma}=l^*\alpha = (l^*\alpha_{\sss A}) l^*(e^{\sss A}) =
    (\alpha_\mu)^{~}_{\sss\Sigma}\, e^\mu.
\end{equation*}
\begin{proposition}
The restricted Laplace-de~Rham operator on the one form: 
$(\square_{n+1}a)_{\sss \Sigma} =\square_{\sss \Sigma} a_{\sss \Sigma} + {}^n\!B_{\sss \Sigma}$, reads
\begin{equation}
    \label{eq:Restriction_Rn+1}
(\square_{n+1}a)_{\sss \Sigma} =\square_{\sss \Sigma} a_{\sss \Sigma} -\epsilon H^2\left[D^2(a_\mu)_{\sss \Sigma} e^\mu +(n-1)D(a_\mu)_{\sss \Sigma} e^\mu+2d_{\sss \Sigma}(i_{\sss D}a)  +(n-2)a_{\sss \Sigma}\right],
\end{equation}
where $d_{\sss \Sigma}= l^*d$.
\end{proposition}
\begin{proof}
In the restriction of $\square_{n+1}a$ 
from Eq.~(\ref{EQ-FormGeneralCoeffStruct}) only the $\mu$-component 
remains, the contribution to $e^n$ being mapped to zero, that is: 
$(\square_{n+1}a)_{\sss\Sigma} = \left[(\square_{n+1}a)_\mu\right]_{\sss\Sigma} e^\mu=: \left[B_\mu e^\mu\right]_{\sss \Sigma}$. 
The r.h.s. of
this last expression divides into two parts: one that contains terms 
with  Greek indexes only, another, that we shall denote ${}^n\!B$, 
contains terms indexed by $n$.
The former leads after restriction to the Laplace-de~Rham operator on $\Sigma_n$ acting on $a_{\sss\Sigma}$: $\square_{\sss\Sigma} 
(a_{\sss\Sigma})$, the latter is calculated in appendix \ref{APP-Calcul-sB}, one finds
\begin{equation}\label{EQ-coeff-sB}
    {}^n\!B= -\epsilon \frac{1}{|y^2|} [D^2(a_\mu)e^\mu +(n-1)D(a_\mu) e^\mu+2 e_\mu(i_{\sss D}a)e^\mu  +(n-2)a_\mu e^\mu].
\end{equation}
Adding this latter term to the former, with $|y^2| = H^{-2}$, proves the formula.
\end{proof}
\begin{remark}
For a transverse one-form ($y^{\sss A} a_{\sss A}=0$) the term
$d_{\sss \Sigma}(i_{\sss D}a)$ vanishes. Also, for $a$ homogeneous of  degree $r$ (thus $a_\mu$ homogeneous of degree $r-1$) the 
above 
expression reduces to
\begin{equation} 
(\square_{n+1}a)_{\sss \Sigma} = \square_{\sss \Sigma} a_{\sss \Sigma} - \epsilon H^2 [(r+1)(r+n-2)a_{\sss \Sigma}+ 2d_{\sss 
\Sigma}(i_{\sss D}a)  ].
\end{equation}
\end{remark}

Eq.~\eqref{eq:Restriction_Rn+1} might be recast in an alternate form containing the Lie derivative along $D$ which provides us the following key formula.

\begin{theorem} The Laplace-de Rham operator acting upon a one-form $a$ restricted from $\setR^{n+1}$ to $\Sigma_n$ reads:
\begin{equation}
    \label{EQ-RestricFormRn+1IntrinsicVersion}
    (\square_{n+1}a)_{\sss \Sigma} = \square_{\sss \Sigma} a_{\sss \Sigma} - \epsilon H^2 
[ {\cal L}^2_{\sss D}(a)+(n-3){\cal L}_{\sss D}(a) + 2d(i_{\sss D}a)  ]_{\sss \Sigma},
\end{equation}
where no reference frame is involved.
\end{theorem}
\begin{proof}
Since ${\cal L}_{\sss D} e^{\sss A}= e^{\sss A}$, thus  implying $D(a_{\sss A})e^{\sss A}= {\cal L}_{\sss 
D}(a)-a$ and $D^2(a_{\sss A})e^{\sss A}= {\cal L}^2_{\sss D}(a)- 2 {\cal L}_{\sss D}(a) +a$, then making these replacements in 
the term ${}^n\!B_{\sss \Sigma}$ recasts Eq.~\eqref{eq:Restriction_Rn+1} into the intrinsic form Eq.~\eqref{EQ-RestricFormRn+1IntrinsicVersion}.
\end{proof}
\begin{remark}[1]
Eq.~\eqref{EQ-RestricFormRn+1IntrinsicVersion} with $\epsilon=-1$ is still valid in the case of the euclidean sphere ($\Sigma_n=S_n$) embedded in the euclidean space $\setR^{n+1}$.
This can be proved through a straightforward adaptation of our proof.
\end{remark}
\begin{remark}[2]
    Thanks to the Weitzenböck  formula, see appendix \ref{APP-FUE}, Eq.~\eqref{EQ-RestricFormRn+1IntrinsicVersion} can be adapted to the Laplace-Beltrami operator:
  \begin{equation*}
    (\Delta_{n+1}a)_{\sss \Sigma} = \Delta_{\sss \Sigma} a_{\sss \Sigma} - \epsilon H^2 
[ {\cal L}^2_{\sss D}(a)+(n-3){\cal L}_{\sss D}(a) + 2d(i_{\sss D}a)  ]_{\sss \Sigma}-r(\sharp a_{\sss \Sigma},\mathrm{Id}),
\end{equation*}  
where $r$ is the Ricci tensor and $r(\sharp a_{\sss \Sigma},\mathrm{Id})$ is the 1-form such that $r(\sharp a_{\sss \Sigma},\mathrm{Id})(u)=r(\sharp a_{\sss \Sigma},u)$.
\end{remark}

\subsubsection{Restriction of the scalar operator}\label{SUBSEC-RestRn+1OfScalar}
For completeness, we note that the scalar operator can be restricted along the same lines as the one-form. The splitting between 
Greek-indexed and $n$-indexed term in Eq.~(\ref{EQ-ScalarGeneralCoeffStruct}) leads again to a first term
which restricts in $\square_{\sss \Sigma} \phi_{\sss \Sigma}$ and 
a second term which, keeping the notation of Sec.~\ref{SUBSEC-RestRn+1OfOneForm}
for simplicity, reads
\begin{equation*}
 {}^n\!B := \eta^{nn}[e_n^2(\phi) + e_n(\phi) c^{\sss A}_{{\sss A} n}].    
\end{equation*}
Thanks to the expressions of $\eta^{nn}$, $e_n$ and $c^{\sss A}_{{\sss A} n}$ 
Sec.~\ref{SUBSEC-AdaptedFramesRn+1}-\ref{SUBSEC-CoeffAnholoRn+1} one 
obtains
\begin{equation*}
 {}^n\!B := -\frac{\epsilon}{|y|^2}\left(D^2 +(n-1)D\right)\phi.    
\end{equation*}
Using the above formula then proves the following.
\begin{theorem}
The restriction of the Laplace-de~Rham operator on a scalar field $\phi$ reads
\begin{equation}\label{EQ-RestricScalarRn+1IntrinsicVersion} 
    (\square_{n+1}\phi)_{\sss \Sigma} 
    = \square_{\sss \Sigma}\phi_{\sss \Sigma} 
    -\epsilon H^2 [ D^2(\phi)+(n-1) D(\phi)]_{\sss \Sigma},
\end{equation}
independently of any frame.
\end{theorem}
Once again this formula, with $\epsilon=-1$ is still valid in the case of the euclidean sphere ($\Sigma_n=S_n$) embedded in the 
euclidean space $\setR^{n+1}$.

\begin{remark}
If $\phi$ is homogeneous of 
degree $r$ the above expression reduces to
\begin{equation*} 
(\square_{n+1}\phi)_{\sss \Sigma} = \square_{\sss \Sigma}\phi_{\sss \Sigma} -\epsilon H^2 r( r+n-1) \phi_{\sss \Sigma}.
\end{equation*}
\end{remark}

\section{Restriction from $\setR^{n+2}$}\label{SEC-Rn+2}
\subsection{Adapted frames of $\setR^{n+2}$}\label{SUBSEC-AdaptedFramesRn+2}

Hereafter, $\setR^{n+2}$ stands for the oriented pseudo-euclidean space, the signature of the metrics being  $(2,n)$. 

\begin{notation}
 In $\setR^{n+2}$, we consider the $n+1$-dimensional 
plane $P$ defined by the equation $f(y)=1$ where $f$ is a homogeneous polynomial  of degree one. We assume that the normal vector field of $P$ is nowhere isotropic and we note  
$F=\sharp df$ this normal vector, and also $F^2=\epsilon H^2$ where $H>0$.
\end{notation}
\begin{proposition}
    The intersection $\Sigma_n$ of the null cone $\mathcal{C}$ of $\setR^{n+2}$ and $P$, endowed with the metrics induced from that of the ambient space, is the de~Sitter or Anti-de~Sitter space, whether $\epsilon=1$ or $\epsilon=-1$,  with radius $H$. The position of $F$, relative to the null cone of $\setR^{n+2}$,  is ``time-like" ($F^2>0$) for de~Sitter space and ``space-like" ($F^2<0$) for Anti-de~Sitter space. In this construction $P$ is identified with the $\setR^{n+1}$ of the previous section $P\simeq\setR^{n+1}$.
\end{proposition}
\begin{proof}
We can, without loss of generality,  consider that the canonical Cartesian coordinates basis of $\setR^{n+2}$ is such that
\begin{equation}\label{EQ-f(y)}
    f(y) = H y^{n+1},\; H > 0,
\end{equation} and that
the metric reads $\eta = \mathrm{diag} (+, -, \ldots, -, -\epsilon, \epsilon)$.
See \cite{Huguet:2006fe} for another choice to control the zero curvature limit.
This implies that
\begin{equation}\label{EQ-FOnCoordn+1}
F= \sharp df = \epsilon H\partial_{n+1}.    
\end{equation}
In effect, written in this coordinate system,  the equations for the intersection of $\mathcal{C}$ and $P$
\begin{equation*}
\left\{
\begin{aligned}
    C(y) &= y^\mu y_\mu - \epsilon (y^n)^2 + \epsilon  (y^{n+1})^2 = 0,\\
    f(y) &= H y^{n+1} = 1,
\end{aligned}
    \right.
\end{equation*}
define $\Sigma_n$ as the de~Sitter, or Anti-de~Sitter space for, respectively, $\epsilon = +1$, and $\epsilon = -1$. 
\end{proof}

Note, for future use, that 
\begin{equation}\label{EQ-D=Dp+Dn+1}
D= y^\alpha \partial_\alpha = D_{\sss P} + D_{n+1}, %\frac{1}{2}\sharp dC
\end{equation}
where $D_{\sss P} := y^{\sss A} \partial_{\sss A}$ is the  dilation operator of $P$, that appears in the $\setR^{n+1}$ reduction Sec.~\ref{SEC-Rn+1}, 
and $D_{n+1} := y^{n+1}\partial_{n+1}$.    

\noindent\textbf{Adapted frame for $\Sigma_n$ in $\setR^{n+2}$.}
Now, a field of  frames adapted to $\Sigma_n$ in $\setR^{n+2}$ is built as follows. 
 We first 
build the adapted frames to $ \Sigma_n$ in $U_{\sss Q}\subset P $: $\{e_\mu, e_n\}$, as described in Sec.~\ref{SUBSEC-AdaptedFramesRn+1}.   Then, one supplements each frame $\{e_\mu, e_n\}$
by the normal to $P$: $\partial_{n+1}$, located at the same point of $P$. Finally, we define the cylinder $W_{\sss Q}$ which is the union of all the ``vertical'' lines directed by $\partial_{n+1}$ and crossing $U_{\sss Q}$. We then extend the frames $\{e_\mu, e_n, \partial_{n+1}\}$ 
of $P$ to $W_{\sss Q}\subset\setR^{n+2}$ by  translating, without change, along the vertical lines. Namely,   for $y\in W_{\sss Q}$ we note $y_{\sss P}$ the orthogonal projection of $y$ on $P$ and \emph{define} $e_{\sss A}^{\sss B}(y)$ through the formula $e_{\sss A}^{\sss B}(y)=e_{\sss A}^{\sss B}(y_{\sss P})=e_{\sss A}^{\sss B}(y^0,\ldots,y^n)$ and, writing for simplicity $e_{n+1}=\partial_{n+1}$,  $e^{\sss A}_{n+1}=e^{n+1}_{\sss A}=0$ and $e_{n+1}^{n+1}=1$. We thus obtained an adapted frame $e_\alpha, \alpha=0,\ldots,n+1$ fulfilling the important properties $D(e_\alpha^\beta)= F(e_\alpha^\beta)=0$, that is: the coefficients are homogeneous of degree 0 and do not depend on $y^{n+1}$. 
This implies that $[e_\mu,\partial_{n+1}] = [e_n,\partial_{n+1}]=0$. 
That is, all anholonomy coefficients 
containing the index $n+1$ are equal to zero.
The adapted frame just built will be used in the following sections.

\subsection{Restriction to (A)dS spaces}

\subsubsection{Restriction of the one-form operator}\label{SUBSEC-RestRn+2OfOneForm}

In the present paragraph we derive the restriction of the one form $\square_{n+2} a$ to the submanifold $\Sigma_n$, the (A)dS
space.
We consider Eq.~(\ref{EQ-FormGeneralCoeffStruct}) as the starting point,
we apply on it the pullback $m^*$ of the canonical injection map $m: \Sigma_n \rightarrow \setR^{n+2}$.  
\begin{notation}
If $\alpha$ is a $p$-form on $\setR^{n+2}$ we set $\alpha_{\sss\Sigma} := m^*\alpha$. 
\end{notation}

From the definition of an adapted frame (see Sec.~\ref{SUBSEC-AdaptedFramesRn+1}) one has $m^* 
e^\mu = e^\mu$ and 
$m^* e^n= m^* e^{n+1} = 0$, then, if $\alpha$ is a 1-form,
\begin{equation*}
    \alpha_{\sss \Sigma}=m^*\alpha = (m^*\alpha_{\alpha}) m^*(e^{\alpha}) =
    (\alpha_\mu)^{~}_{\sss\Sigma}\, e^\mu.
\end{equation*}

\begin{proposition}
The restricted Laplace-de~Rham operator on the one form $(\square_{n+2}a)_{\sss \Sigma}$ reads
\begin{align}
\label{EQ-RestricFormRn+2ExtendedFrameVersion}
m^*(\square_{n+2}a)= \square_{\sss \Sigma} a_{\sss \Sigma}+m^*\Bigl[&\epsilon(\partial^2_{n+1} a_\mu)e^\mu \\
&- \epsilon H^2 
\left({\cal L}^2_{\sss D_{\sss P}}(a)+(n-3){\cal L}_{\sss D_{\sss P}}(a)+ 2\,d(i_{\sss D_{\sss P}}a)\right)\Bigr].\notag
\end{align}
\end{proposition}
\begin{proof}
The calculation parallels that of Sec.~\ref{SUBSEC-RestRn+1OfOneForm}.
The r.h.s. of the pullback by $m$ of 
Eq.~(\ref{EQ-FormGeneralCoeffStruct}) divides now 
into three parts: the first one, containing only Greek indexes $\mu, \nu, \ldots$, leads after restriction to the Laplace-de~Rham 
operator on $\Sigma_n$ acting on 
$a_{\sss\Sigma}$: $\square_{\sss\Sigma} (a_{\sss\Sigma})$, the second one denoted ${}^{n}\!B$ as in Sec.~\ref{SUBSEC-RestRn+1OfOneForm},
contains terms indexed by $n$
but not by $n+1$, the last one contains terms indexed by $n+1$, and since anholonomy  coefficients involving $n+1$ are equal to 
zero, it reduces to 
$m^* \epsilon(\partial^2_v a_\mu)e^\mu$. The calculation of  ${}^{n}\!B$ is the same as in the $\setR^{n+1}$ embedding space Sec.~\ref{SUBSEC-RestRn+1OfOneForm}, 
one obtains thus Eq.~(\ref{EQ-coeff-sB})
with the replacements $D \rightarrow D_{\sss P}$ and $y^2 \rightarrow y^2_{\sss P} := y^{\sss A}y_{\sss A}$, that accounts for the embedding space $\setR^{n+2}$ in place of $\setR^{n+1}$.
\end{proof}

Finally, the second term of the r.h.s. of this equation can be recast under an intrinsic form, that is without reference to 
specific frames.

\begin{theorem} The Laplace-de Rham operator acting upon the one-form a restricted from
$\setR^{n+2}$ to $\Sigma_n$ reads:
\begin{align}\label{EQ-RestricFormRn+2IntrinsicVersion}
m^*(\square_{n+2}a) = \square_{\sss \Sigma} a_{\sss \Sigma}
+ m^*&\left[ -\epsilon H^2\left( 
\mathcal{L}_{\sss D}^2 + (n-3)\mathcal{L}_{\sss D} + 2 d i_{\sss D}\right) 
+ 2\mathcal{L}_{\sss F}\mathcal{L}_{\sss D} \right.\nonumber\\
&\ \left. + (n - 4) \mathcal{L}_{\sss F} + 2 d i_{\sss F}\right](a)\nonumber\\
= \square_{\sss \Sigma} a_{\sss \Sigma}
+ m^*&\left[ -\epsilon H^2\left( 
\mathcal{L}_{\sss D}^2 + (n-3)\mathcal{L}_{\sss D} + 2 d i_{\sss D}\right) 
+ 2\mathcal{L}_{\sss D}\mathcal{L}_{\sss F} \right.\nonumber\\
&\  \left. + n \mathcal{L}_{\sss F} + 2 d i_{\sss F}\right](a),
\end{align}
\end{theorem}
\begin{proof}
The details of the calculation can be found in appendix \ref{APP-IntrinsicFormOfEqVect}.
Switching between both expressions makes use of  $[F,D]=2F$,  due to the homogeneity of $f$.\end{proof}

\begin{remark}
    Thanks to the Weitzenböck formula, see appendix \ref{APP-FUE}, Eq.~\eqref{EQ-RestricFormRn+2IntrinsicVersion} can be adapted to the Laplace-Beltrami operator:
  \begin{align*}
    (\Delta_{n+2}a)_{\sss \Sigma} = \Delta_{\sss \Sigma} a_{\sss \Sigma} + m^*&\left[ -\epsilon H^2\left( 
\mathcal{L}_{\sss D}^2 + (n-3)\mathcal{L}_{\sss D} + 2 d i_{\sss D}\right) 
+ 2\mathcal{L}_{\sss F}\mathcal{L}_{\sss D} \right.\\
&\ \left. + (n - 4) \mathcal{L}_{\sss F} + 2 d i_{\sss F}\right](a)-r(\sharp a_{\sss \Sigma},\mathrm{Id}),
\end{align*}  
where $r$ is the Ricci tensor and $r(\sharp a_{\sss \Sigma},\mathrm{Id})$ is the 1-form such that $r(\sharp a_{\sss \Sigma},\mathrm{Id})(u)=r(\sharp a_{\sss \Sigma},u)$.
\end{remark}

\subsubsection{Restriction of the scalar operator}\label{SUBSEC-RestRn+2OfScalar}

Let us consider again the scalar case. 
Eq.~(\ref{EQ-ScalarGeneralCoeffStruct}) for a, not necessarily homogeneous, 
scalar field $\phi$
splits, as in the one-form case, into three terms corresponding to those containing Greek indexes only, those involving the index 
$n$, and those 
indexed by $n+1$. The part relating to Greek indexes only, leads, upon reduction, to the scalar Laplace-de~Rham operator on 
$\Sigma_n$, the (A)dS space, 
the others are calculated in a straightforward way in appendix \ref{APP-IntrinsicFormOfEqScal}.
This proves the following formula.
\begin{theorem}
The Laplace-de Rham operator acting upon a scalar field $\phi$ restricted from $\setR^{n+2}$ to $\Sigma_n$ reads:
\begin{align}\label{EQ-RestricScalarRn+2IntrinsicVersion}
    m^*(\square_{n+2}\phi) &= \square_{\sss \Sigma}\phi_{\sss \Sigma} -\epsilon H^2\left[D^2\phi +(n-1)D\phi\right]_{\sss \Sigma} 
    +\left[2FD\phi+(n-2)F\phi\right]_{\sss \Sigma}
    \nonumber\\
    &=\square_{\sss \Sigma}\phi_{\sss \Sigma} -\epsilon H^2\left[D^2\phi +(n-1)D\phi\right]_{\sss \Sigma} 
    +\left[2DF\phi+(n+2)F\phi\right]_{\sss \Sigma}.%\label{EQ-RestricScalarRn+2IntrinsicVersion}.
\end{align}
\end{theorem}
\begin{remark}
For $\phi$ homogeneous of degree $r$, the above expression reduces to
\begin{equation*}
    (\square_{n+2}\phi)_{\sss \Sigma}= \square_{\sss \Sigma}\phi_{\sss \Sigma} -\epsilon H^2 r(r+n-1)\phi_{\sss \Sigma} + (2r + 
    n-2)F(\phi)_{\sss \Sigma}.
\end{equation*}
Note that for $n=4$ we retrieve our previous result (Eq.~(9) of \cite{Zapata:2017gqg}) as a particular case.
\end{remark}

\section{The additional terms}\label{SEC-AdditionnalTerms}

We observe that all restricted operators we obtained in previous sections 
Eqs.~(\ref{EQ-RestricFormRn+1IntrinsicVersion}, \ref{EQ-RestricScalarRn+1IntrinsicVersion},  
\ref{EQ-RestricFormRn+2IntrinsicVersion},
\ref{EQ-RestricScalarRn+2IntrinsicVersion}), share the common structure:
\begin{equation}\label{EQ-GenericEqWith-AT}
    (\square_d \alpha)_{\sss \Sigma}=\square_{\sss \Sigma}\alpha_{\sss \Sigma}+\mbox{AT}(\alpha), 
\end{equation}
where $\alpha$ is a one-form or a scalar field and AT stands for ``Additional Term".

A natural question is whether the Additional Term has some universal characteristic or can be anything. 
In order 
to address this question,
we are interested in the converse situation in which one considers as a starting point the expression
\begin{equation}\label{EQ-GenericTermWith-X}
 \square_{\sss \Sigma}\beta + \chi(\beta),   
\end{equation}
where $\beta$ and $\chi$ live on $\Sigma_n$ and are respectively a one-form or scalar field, and an arbitrary smooth operator. 
We will show that it is possible to choose an extension, to the embedding space, of $\beta$ such that the Additional Term 
corresponding to this extension is equal to $\chi$.
Such an extension of $\beta$ will be denoted $\varpi(\beta)$ where $\varpi$  is a right inverse (sometimes called a \emph{section}) of $m^*$: 
$m^*\circ\varpi = \mbox{Id}$. The answer to this question is given by the following theorem.
\begin{theorem}
Let $\beta$ be a 1-form or a scalar field on $\Sigma_n$ and $\chi$ a smooth operator, then there exists an extension $\varpi_\chi$ such that
\[
    (\square_{n+2} \varpi_\chi(\beta))_{\sss \Sigma}=\square_{\sss \Sigma}\beta+\chi(\beta).
\]
\end{theorem}
In other words, any additional term $\chi$ on the embedded pseudo-sphere $\Sigma_n$ can be obtained by the choice of an appropriate extension of the one-form or scalar field. One can see easily that proving this theorem is equivalent to finding an extension $\varpi_\chi$ such that ${\rm AT}\circ\varpi_\chi=\chi$. The following subsections are devoted to building explicitly these extensions.

Note that the same result holds in $\setR^{n+1}$ context and, for technical reasons, we begin with addressing this case.

\subsection{Additional terms in the $\setR^{n+1}$ context}
We recall that $l$ is the canonical injection of $\Sigma_n$ into $\setR^{n+1}$ and its pullback $l^*$ the corresponding restriction for 
scalar and one-forms.

\begin{definition}
Let  $y\in\setR^{n+1}$ be a point such that $\sgn(y^{\sss A}y_{\sss A})=-\epsilon$, let $\rho$ be a map over $\Sigma_n$, and let $r\in\setR$.
Then, $\rho$ is \emph{extended by homogeneity} from $\Sigma_n$ to $\setR^{n+1}$ by setting 
\begin{equation}\label{EQ-ExtensionHomogeneity-Def-rho(r)}
\rho^{(r)}(y) = \rho\left(\frac{y}{H\sqrt{|y_{\sss A}y^{\sss A}}|}\right)\left(H\sqrt{|y_{\sss A}y^{\sss A}|}\right)^r.
\end{equation}
\end{definition}
\begin{remark}
    The condition $\sgn(y^{\sss A}y_{\sss A})=-\epsilon$ ensures that the half-line emanating from the center of $\setR^{n+1}$ and containing $y$ intercepts $\Sigma_n$.
\end{remark}

Since $y/(H\sqrt{|y_{\sss A}y^{\sss A}|})\in\Sigma_n$, then $\rho^{(r)}$ is homogeneous of degree $r$ and $(\rho^{(r)})_{\sss \Sigma}=\rho$ by construction.
In regards to the one-forms, one can define the extension by
homogeneity of order $s$ for any one-form $h=h_\mu e^\mu$ on $\Sigma_n$ through 
\begin{equation*}
    h^{(s)}= h_\mu^{(s-1)}e^\mu.
\end{equation*}

Remark that, since $D$ is along the (outer) normal to $\Sigma_n$,  
any one-forms $k = k_\mu e^\mu$ on $E$ satisfy the transversality condition $i_{\sss D} k = 0$. This is the case for $h^{(s)}$ 
which is thus an extension of $h$, transverse and homogeneous of degree $s$.
From Eq.~(\ref{EQ-RestricFormRn+1IntrinsicVersion}) we see that
AT, in the present situation, involves the Lie derivative ${\mathcal L}_{\sss D}$. Let us determine its action on $h^{(s)}$, one 
has successively
\begin{align*}
    {\mathcal L}_{\sss D}(h^{(s)})
    &=i_{\sss D}dh_\mu^{(s-1)}e^\mu +d\ i_{\sss D}h^{(s)}\\
    &=D(h_\mu^{(s-1)})e^\mu +h_\mu^{(s-1)}i_{\sss D}de^\mu\\
    &=(s-1)h_\mu^{(s-1)}e^\mu-\frac{1}{2}h_\mu^{(s-1)}
    \sqrt{|y^2|}i_nc^\mu_{\sss AB}e^{\sss A}\wedge e^{\sss B} \\
    &=(s-1)h^{(s)}-\frac{1}{2}h_\mu^{(s-1)}
    \sqrt{|y^2|}c^\mu_{\sss AB}(\delta^{\sss A}_n e^{\sss B}
    -\delta^{\sss B}_n e^{\sss A})\\
    &=(s-1)h^{(s)}-h_\mu^{(s-1)}
    \sqrt{|y^2|}c^\mu_{n{\sss A}}e^{\sss A}\\
    &=s h^{(s)},
    \end{align*}
where transversality has been used to eliminate the second term on the r.h.s. in the first line.

\begin{proposition}\label{THM-section-Rn+1_form}
    Let $b := \beta$ in Eq.~\eqref{EQ-GenericTermWith-X} be a one-form field on $\Sigma_n$, then a section $\varpi_\chi$ such that $\mbox{AT}\circ\varpi_\chi(b)= \chi(b)$ is given by
    \begin{equation*}
        \varpi_\chi(b)=b^{(0)}-\epsilon H^{-2} (\chi^{(s)}-\chi^{(0)}),  
    \end{equation*}
    with $s$ a solution of $S(s)=1$ where $S(X)=X^2+(n-3)X$.
\end{proposition}
\begin{proof}
Note first, that applying $l^*$ to both members of the above relation shows
that $\varpi$ is a section of $l^*$. Then, starting from 
Eq.~(\ref{EQ-RestricFormRn+1IntrinsicVersion}),
we obtain: 
\begin{equation*}
{\rm AT} =- \epsilon H^2 l^*\circ
[ S({\cal L}_{\sss D}) + 2d(i_{\sss D})  ],
\end{equation*}
and
\begin{align*}
    \mbox{AT}(\varpi_\chi(b))&= -\epsilon H^2 l^*S({\cal L}_{\sss D})\varpi_\chi(b)\\
    &=-\epsilon H^2 l^*S({\cal L}_{\sss D})[b^{(0)}-\epsilon H^{-2}
(\chi^{(s)}(b)-\chi^{(0)}(b))]\\
    &= 0+l^*S({\cal L}_{\sss D})
    (\chi^{(s)}(b))\\
&=S(s)l^*(\chi^{(s)}(b))\\
&=\chi(b),
\end{align*}
where the transversality of $\chi^{(s)}(b)$ has been used in the first and the third line.
\end{proof}

The scalar case can be treated in a similar way and we can prove the following proposition in a similar manner.
\begin{proposition}\label{THM-section-Rn+1_scal}
Let $\psi:=\beta$ be a scalar field on $\Sigma_n$, and let $\varpi_\chi$ be defined through
\begin{equation*}
 \varpi_\chi(\psi)=\psi^{(0)}-\epsilon H^{-2}(\chi^{(s)}(\psi)-\chi^{(0)}(\psi)),  
\end{equation*}
where $s$ is a solution of the equation  $Q(s):=s^2+(n-1)s=1$.
Then $\varpi_\chi$ is a section of $l^*$, that is it fulfills $l^*\circ \varpi_\chi=\mbox{Id}$,
which satisfies ${\rm AT}\circ\varpi_\chi=\chi$.
\end{proposition}

\subsection{Additional terms in the $\setR^{n+2}$ context}
As in the previous section, we have to find $\varpi'_\chi$ a section of $m^*$ such that ${\rm AT}'\circ\varpi'_\chi=\chi$ where  
${\rm AT}'$ is the additional term in 
$\setR^{n+2}$ context, see Eq.~\eqref{EQ-RestricFormRn+2IntrinsicVersion}, namely
\begin{equation*}
{\rm AT}'= m^*\left[ -\epsilon H^2 (
S(\mathcal{L}_{\sss D})+ 2 d i_{\sss D}) 
+ 2\mathcal{L}_{\sss D}\mathcal{L}_{\sss F} 
+ n \mathcal{L}_{\sss F} + 2 d i_{\sss F}\right].
\end{equation*}
Let $p$ be the orthogonal projection of $\setR^{n+2}$ onto $P\sim \setR^{n+1}$ (remember that the normal vector of $P$ is non-isotropic), one can verify that $l=p\circ m$
and, as a consequence, $p^*\circ\varpi_\chi$ is a section of $m^*$.

\begin{proposition}
    Let $p$ be the orthogonal projection of $\setR^{n+2}$ onto $P\sim\setR^{n+1}$, as above, let $\varpi_\chi$ be a section in $\setR^{n+1}$, as in Proposition~\ref{THM-section-Rn+1_form} or \ref{THM-section-Rn+1_scal}.
    Then $\varpi'_\chi:=p^*\circ\varpi_\chi$ fulfills ${\rm AT}'\circ\varpi'_\chi=\chi$.
\end{proposition}

Before proving that, remarks regarding $p^*$ are in order. Let $h=h_{\sss A}e^{\sss A}$ be a one-form on $P\sim \setR^{n+1}$, then we have
\begin{equation*}
(p^*h)(y^0,\ldots,y^n,y^{n+1})= h_{\sss A}(y^0,\ldots,y^n)e^{\sss A}.
\end{equation*}
This proves several facts. First ${\mathcal L}_{\sss F}\circ p^*=i_{\sss F}\circ p^*=0$, second $p^*h$ is homogeneous of degree $s$ as soon as $h$ is, and finally $i_{\sss D}p^*h=0$ as soon as $i_{\sss D}h=0$.
\begin{proof}
    Taking into account the properties of $p^*$, in the case of the one-form field we have
\begin{align*}
{\rm AT}'\circ \varpi'_\chi&=m^*\left[ -\epsilon H^2 (
S(\mathcal{L}_{\sss D})+ 2 d i_{\sss D}) 
+ 2\mathcal{L}_{\sss D}\mathcal{L}_{\sss F} 
+ n \mathcal{L}_{\sss F} + 2 d i_{\sss F}\right]\circ p^*\circ\varpi_\chi\\
&=-\epsilon H^2m^*S(\mathcal{L}_{\sss D})\circ p^*\circ\varpi_\chi\\
&=-\epsilon H^2m^*S(s)p^*\circ\varpi_\chi\\
&=-\epsilon H^2l^*\circ\varpi_\chi\\
&=\chi.
\end{align*}
The scalar case can be treated in the same way: replacing $\varpi_\chi$ with $p^*\circ\varpi_\chi$  leads to the same result.
\end{proof}

We  again see that any operator $\Box_\Sigma b + \chi(b)$ on $\Sigma_n$ can be reached as a restriction  of the operator 
$\Box_{n+2}a$ where $a$ is an appropriate extension of $b$.

\section{Conclusion}\label{SEC-Conclusion}
The keystone of our construction is the existence of an adapted frame for which the commutation relations, leading to 
anholonomic coefficients, are easily calculated. This adapted frame is deeply related to the pseudo-sphere and the 
maximal symmetry of (A)dS. 
Other space-times are no longer maximally symmetric. Nevertheless, many of them share the same conformal group SO$(2,n)$, 
including Friedmann-Lemaître-Robertson-Walker spaces. 
This should allow us, in a future work, to generalize to these spaces the results obtained here.
\vfill

\appendix
\section{Details on calculations}
\subsection{Overview of the main notions of differential geometry used in this work}\label{APP-FUE}

We gather and comment here some definitions and formulas from differential geometry frequently used in this paper (the notations are those of \cite{Fecko:2006}).
We also relate, using our notations, the Laplace-de Rham to the Laplace-Beltrami operators through the 
Weitzenböck formula.

We consider an oriented pseudo-Riemanian  $d$-dimensional manifold $(M,g)$ and the set $\Lambda(M)=\oplus_p\Lambda^p(M)$ of differential forms on $M$. 
We note $\langle \alpha,v\rangle$ the natural pairing between a 1-form $\alpha\in\Lambda^1(M)$ and a vector $v\in T(M)$. We let, as usual, $\sharp\alpha$ be the vector such that $g(\sharp\alpha,v)=\langle\alpha,v\rangle$ and
$\flat v$ be the 1-form such that $\langle\flat v,u\rangle=g(v,u)$ for any vector $u$. We often use the notation ~$\widetilde{ }$~  instead of both $\sharp$ and $\flat$. We note also $\hat\eta$ the linear operator of $\Lambda(M)$ defined through: $\hat\eta\alpha=(-1)^p\alpha$ {\it iff} $\alpha\in \Lambda^p(M)$. 

Beside the well-known interior product $i_v$ we define the corresponding creator $j_v$  through $j_v\alpha=\tilde v\wedge \alpha$ where $v\in T(M)$.
The basic relationship between them is
\[ i_uj_v+j_vi_u= g(u,v).\] 
The interior product relates to the Lie derivatives through the Cartan formula:
\[{\mathcal L}_{v}=i_{v}d+di_{v}.\]
In this paper we make intensive use of positively oriented orthonormal frames $e_a$ verifying $g(e_a, e_b)=\eta_{ab}$ and their dual basis $e^a$. In such a frame, the volume form reads $\omega=e^1\wedge\cdots\wedge e^d$. We note $i^a=\eta^{ab}i_{e_b}$ and $j^a=\eta^{ab}j_{e_b}$ and obtain the very useful algebra:
\begin{align*}
    &\eta^{ab}= i^aj^b  + j^bi^a,\\
    &i^ai^b=-i^bi^a,\\
    &j^aj^b=-j^bj^a.
\end{align*}

In order to introduce the Hodge operator, we begin with the inner product of $1$-forms $(\alpha,\beta)_g=g(\tilde\alpha,\tilde\beta)$ that we extend to the Gram inner product of two $p$-forms $\alpha_1\wedge\cdots\wedge\alpha_p$ and $\beta_1\wedge\cdots\wedge\beta_p$, where $\alpha_i$ and $\beta_i$ are $1$-forms, through
\[(\alpha_1\wedge\cdots\wedge\alpha_p,\beta_1\wedge\cdots\wedge\beta_p)_g= \det M, %\det(\alpha_i,\beta_j)_g.
\]
where $M$ is the $p\times p$ matrix whose entries are $M_{ij} = (\alpha_i,\beta_j)_g$.

We can now introduce the Hodge operator $\alpha\mapsto\ast\alpha,\ \Lambda^p(M)\to\Lambda^{d-p}(M)$ such that, for any $p$-form $\beta$,
\begin{equation*}
    \beta\wedge\ast\alpha=(\beta,\alpha)_g\omega.
\end{equation*}
This operator fulfills the elementary properties
\begin{align*}
    &\ast^{-1} =\sgn(g)  \ast  \hat\eta^{d+1},\\ 
    &\ast \omega= \sgn(g),\\ 
    &\ast 1 = \omega,
\end{align*}
and interacts with the $i$ and $j$ operators in the following way
\begin{align*}
    &\ast i_v= -j_v\ast \hat\eta,\\ %\mbox{ and }
    &\ast j_v=i_v\ast\hat\eta,
\end{align*}
from which we obtain the useful formulas
\begin{align*}
    &\ast(e^a\wedge\cdots\wedge e^b)= i^b\ldots i^a \omega,\\% \mbox{ and }
    &\ast i^a\ldots i^b\omega=\sgn(g)\hat\eta^{d+1}e^b\wedge\cdots\wedge e^a.
\end{align*}

Moreover, the $\ast$-operator allows us to define the co-differential operator $\delta$ on $\Lambda(M)$  through
\[ \delta =\ast^{-1} d\ast \hat\eta,\]
which, in the case of 1-forms,
relates to the usual divergence through
$\delta\alpha=-\mbox{div}(\tilde\alpha)$.
We can now define the Laplace-de~Rham operator $\Box$, which we use intensively in this paper, through
\[\Box=-(d\delta+\delta d).\]

This operator is related to the usual Laplace-Beltrami operator.
Indeed, if $\nabla$ is the Levi-Civita connection associated to $g$ then the second covariant derivative 
$\nabla^2_{uv}=\nabla_u\nabla_v-\nabla_{\nabla_uv}$, where $u$ and $v$ are tangent vectors,  brings about the Laplace-Beltrami operator through its trace, namely:
\[ \Delta = \mbox{tr}\nabla^2=g^{ab}(\nabla_a\nabla_b-\nabla_{\nabla_ae_b}).\]
The link between the Laplace-de Rham $\square$ and the Laplace-Beltrami $\Delta$ operators on $p$-forms is given by
 the Weitzenböck formula, which we recall here with a short proof written in our notations. 
\begin{theorem}[Weitzenböck formula]
Let $\alpha$ a $p$-form and $e_a$ any frame, then 
\begin{equation}\label{EQ-Weitzenbock}
    \square\alpha=\Delta\alpha
    +j^ai^bR(e_a,e_b)\alpha,
\end{equation}
where $R(u,v)=\nabla_u\nabla_v-\nabla_v\nabla_u-\nabla_{[u,v]}$ is the curvature operator and $\nabla$ the Levi-Civita  connection.
\end{theorem}
\begin{proof}
  We first recall the link between the operators $d$, $\delta$ and the Levi-Civita connection $\nabla$:
\begin{align*}
    &\delta=-i^a\nabla_a,\\ 
    &d=j^a\nabla_a,
\end{align*}
in any frame $e_a$. 
We now choose a point $p$ of the manifold and a normal (geodesic) frame centered at $p$. As a consequence, in this frame, $\Gamma_{ab}^c(p)=0$ and $g_{ab}(p)=\eta_{ab}$. We can now work at this point  and in this frame and obtain 
\begin{align*}
    d\delta\alpha&=-j^a\nabla_ai^b\nabla_b\alpha\\
    &=-j^ai^b\nabla_a\nabla_b\alpha,
\end{align*}
and
\begin{align*}
    \delta d\alpha&= -i^b\nabla_bj^a\nabla_a\alpha\\
    &= -i^bj^a\nabla_b\nabla_a\alpha\\
     &= +j^ai^b\nabla_b\nabla_a\alpha-\eta^{ab}\nabla_b\nabla_a\alpha.
\end{align*}
This proves that, in this frame at the point $p$,
\[\square\alpha= \eta^{ab}\nabla_b\nabla_a\alpha+j^ai^b(\nabla_a\nabla_b
-\nabla_b\nabla_a)\alpha.
\]
Moreover, in this frame, $\Delta=\eta^{ab}\nabla_a\nabla_b$ and $R(e_a,e_b)=(\nabla_a\nabla_b
-\nabla_b\nabla_a)$, and we obtain that the formula \ref{EQ-Weitzenbock} is correct in this frame and at this point. However $R$ is tensorial in each entries and then the formula is independent of the frame.% This ends the proof.
\end{proof}
This formula takes a simpler form in two special cases.

\noindent{\bf Scalar case ($p=0$).}
Since $R$ vanishes when applied on scalar field we get 
\begin{equation}
    \square\alpha=\Delta\alpha.
\end{equation}

\noindent{\bf 1-form case ($p=1$).}
Since $[\nabla,\sharp]=0$ we have
\begin{align*}
    j^ai^bR(e_a,e_b)\alpha
    &=j^ai^b\sharp\sharp R(e_a,e_b)\alpha_c e^c\\
     &=j^ai^b\sharp R(e_a,e_b)\sharp \alpha_c e^c\\
      &=j^ai^b\alpha^c\sharp R(e_a,e_b)  e_c\\
      &=j^ai^b\alpha^c\sharp e_dR^d_{\ cab}\\
       &=-\alpha^cj^ai^b\eta_{de} e^eR^d_{\ cba}\\
        &=-\alpha^cj^a\eta^{be}\eta_{de} R^d_{\  cba}\\
        &=-\alpha^cj^a\delta^b_d R^d_{\ cba}\\
        &=-\alpha^cj^a R^b_{\ cba}\\
         &=-\alpha^cj^a r_{ca}\\
         &=-j^a r(\sharp\alpha,e_a),
\end{align*}
with $r$ the Ricci tensor.
In this case the Weitzenböck formula becomes
\begin{equation*}
    \square\alpha=\Delta\alpha
    -j^ar(\sharp\alpha,e_a),
    \end{equation*}
in other words
\begin{equation*}
    \square\alpha(u)=\Delta\alpha (u)
    -r(\sharp\alpha,u).
\end{equation*}

\subsection{Calculation of $\delta(e^a)$}\label{APP-CalcDelta(e^a)}
\begin{align*}
    \delta(e^a)&=-\sgn(g)\ast d\ast e^a\\
    &=-\sgn(g)\ast d i^a\omega\\
    &=-\sgn(g)\eta^{aa}\ast d (-1)^{a-1}   e^1\wedge\cdots\wedge\overset{\lor}{e^a}\wedge\cdots\wedge e^d\\
    &= \sgn(g)(-1)^a\eta^{aa}\ast J,
\end{align*}
with $J =d( e^1\wedge\cdots\wedge\overset{\lor}{e^a}\wedge\cdots\wedge e^d)$.
Thus, using exceptionally the symbol $\sum$ to indicate summations,  one has successively

\begin{align*}
    J &=\sum_{p < a} (-1)^{p-1} e^1\wedge \cdots\wedge de^p\wedge\cdots\wedge\overset{\lor}{e^a}\wedge\cdots\wedge e^d\\
    &\;\;\; +\sum_{p> a} (-1)^{p} e^1\wedge \cdots\wedge\overset{\lor}{e^a}\wedge\cdots\wedge de^p\wedge\cdots \wedge e^d\\
    &=\sum_{p < a} (-1)^{p} e^1\wedge \cdots\wedge (\sum_{b<c}c^p_{bc}e^b\wedge e^c)\wedge\cdots\wedge\overset{\lor}{e^a}\wedge\cdots\wedge e^d\\
    &\;\;\; + \sum_{p> a} (-1)^{p+1} e^1\wedge \cdots\wedge\overset{\lor}{e^a}\wedge\cdots\wedge (\sum_{b<c}c^p_{bc}e^b\wedge e^c)\wedge\cdots\wedge e^d\\
    &=\sum_{p < a} (-1)^{p} e^1\wedge \cdots\wedge (c^p_{pa}e^p\wedge e^a)\wedge\cdots\wedge\overset{\lor}{e^a}\wedge\cdots\wedge e^d\\
    &\;\;\; +\sum_{p> a} (-1)^{p+1} e^1\wedge \cdots\wedge\overset{\lor}{e^a}\wedge\cdots\wedge (c^p_{ap}e^a\wedge e^p)\wedge\cdots\wedge e^d\\
    &=\sum_{p < a} (-1)^{p} (-1)^{p-a+1}c^p_{pa}\omega+\sum_{p> a} (-1)^{p+1}(-1)^{a+1-p} c^p_{ap}\omega\\
    &= (-1)^{a+1}\omega\sum_{p\neq a} c^p_{pa}.
\end{align*}
Finally, returning to the Einstein summation convention, one gets
\begin{equation*}\label{EQ-Delta(e^a)}
    \delta(e^a)= -\eta^{ab}c^p_{pb}.
\end{equation*}
\vfill

\subsection{Derivation of Eq.~(\ref{EQ-ScalarGeneralCoeffStruct})}\label{APP-CalcEqBoxPhiRepOrth}
The expression Eq.~(\ref{EQ-ScalarGeneralCoeffStruct}) for the scalar field is obtained through the following steps:
\begin{align*}
\square \phi
&=-\delta d\phi\\
&=-\delta e_a(\phi) e^a\\
&=-\ast^{-1} d\ast \hat\eta\  e_a(\phi) e^a\\
&=-\sgn(g) \ast \hat\eta^{d+1} d \ast (-1)e_a(\phi)  e^a\\
&=\sgn(g)  \ast  (-1)^{d(d+1)} d e_a(\phi) i^a\omega\\
&=\sgn(g)  \ast  ( d (e_a(\phi))\wedge i^a\omega  + e_a(\phi)d i^a\omega)  \\
&=\sgn(g)  \ast  \left( e_b (e_a(\phi)) e^b\wedge i^a\omega  +e_a(\phi) {\cal L}_{\widetilde {e^a}} (\omega) \right) \\
&=\sgn(g)  \ast  \left( e_b (e_a(\phi)) j^b i^a\omega  +e_a(\phi) {\cal L}_{\widetilde {e^a}} (\omega) \right) \\
&=\sgn(g)  \ast  \left(\eta^{ab} e_b (e_a(\phi))\omega - e_a(\phi) \delta (e^a) \omega \right) \\
&= \eta^{ab} (e_b (e_a(\phi))  - e_a(\phi) \delta (e^a)\\
&= \eta^{ab}[e_ae_b(\phi)+e_a(\phi) c^p_{pb}],
\end{align*}
where the expression of the term $\delta(e^a)$, derived in appendix \ref{APP-CalcDelta(e^a)}, has been used in the last line.

We note for future reference that $\delta \alpha$, $\alpha$ being a one-form, is obtained by the replacement of 
$e_a(\phi)$ by $\alpha_a$ in the above calculation leading to
\begin{equation}\label{EQ-DeltaOneForm}
 \delta(\alpha) = -\eta^{ab}e_a(\alpha_b) + \alpha_a\delta(e^a).
\end{equation}

\subsection{Derivation of Eq.~(\ref{EQ-FormGeneralCoeffStruct})}\label{APP-CalcEqBoxaRepOrth}
In order to express $\square a$ in an ortho-normal basis. 
We first compute the Maxwell operator:
\begin{equation*}
    -\delta d(a_ae^a)
    = \underset{I}{\underbrace{-\delta d(a_a)\wedge e^a}}
    + \underset{II}{\underbrace{ -\delta a_cde^c}}.
\end{equation*}
Calculation of $I$:
\begin{align*}
I&=-\delta e_b(a_a)e^b\wedge e^a\\
&= (-1)^d\sgn(g) \ast d \ast e_b(a_a)e^b\wedge e^a\\
&= (-1)^d \sgn(g)\ast d  e_b(a_a)i^ai^b\omega\\
 &= (-1)^d\sgn(g)\left( \underset{I_1}{\underbrace{\ast d  (e_b(a_a))i^ai^b\omega}}+ \underset{I_2}{\underbrace{\ast e_b(a_a)di^ai^b\omega}} \right).
\end{align*}
For the term $I_1$ one has successively:
\begin{align*}
I_1&= \ast d  (e_b(a_a))i^ai^b\omega\\
&= \ast e_c( (e_b(a_a)))j^ci^ai^b\omega\\
&= \ast e_c( (e_b(a_a)))(\eta^{ca}-i^aj^c)i^b\omega\\
&= e_a( (e_b(a^a)))\ast i^b\omega -e_c( (e_b(a_a)))\ast i^a   \eta^{cb}\omega\\
&=\sgn(g) (-1)^{d+1}\left(e_a( (e_b(a^a)))e^b-\eta^{cb}e_c( (e_b(a_a)))e^a\right)\\
&=\sgn(g) (-1)^{d+1}(e_ae_b(a^a)e^b-\eta^{cb}e_ce_b(a_a)e^a).
\end{align*}
For the term $I_2$ one has successively:
\begin{align*}
I_2&=  \ast e_b(a_a)di^ai^b\omega\\
&= \ast e_b(a_a)({\cal L}_{\tilde{e^a} } - i^ad) i^b\omega\\
&=  \ast e_b(a_a) \left[{\cal L}_{\tilde{e^a} }  ( \tilde{e^b}\rfloor \omega)- i^adi^b\omega\right]\\
&=  \ast e_b(a_a) \left[
    {\cal L}_{\tilde{e^a} }  ( \tilde{e^b})\rfloor \omega
    + \tilde{e^b}\rfloor {\cal L}_{\tilde{e^a} } (\omega)
    - i^a {\cal L}_{\tilde{e^b} } \omega
    \right]\\
&= \ast e_b(a_a) \left[
     [\tilde{e^a},\tilde{e^b}]\rfloor \omega 
    -\delta(e^a) i^b \omega
    + \delta(e^b) i^a \omega\right]\\
%&=   e_b(a_a) \left[-\delta(e^b)\ast i^a \omega- \ast [\eta^{ac} e_c,\eta^{bf}e_f]\rfloor \omega +\ast i^b\delta(e^a) (\omega)\right]\\
&= \ast e_b(a_a) \left[
    \eta^{ac}\eta^{bm}c^l_{cm}\eta_{lf} i^f\omega
    - \delta(e^a) i^b \omega
    + \delta(e^b) i^a \omega\right]\\
%&=   e_b(a_a) \left[-\delta(e^b)\ast i^a \omega- \eta^{ac}\eta^{bm}\ast i_{ [e_c,e_m] }\omega +\delta(e^a) \ast i^b(\omega)\right]\\
&=   (-1)^{d+1} \sgn(g)e_b(a_a) \left[\delta(e^b)e^a+ \eta^{ac}\eta^{bm} \eta_{lf} c^l_{cm} e^f-\delta(e^a) e^b\right]\\
&= (-1)^{d+1} \sgn(g)e_b(a_a) \left[
        -\eta^{bd}c^p_{pd}e^a
        + \eta^{ac}\eta^{bm} \eta_{lf} c^l_{cm} e^f
        + \eta^{ad}c^p_{pd}e^b
    \right] \\
    &= (-1)^{d+1} \sgn(g) \left[
        - \eta^{bd}c^p_{pd} e_b(a_a)
        + \eta^{bm} \eta_{na} c^n_{cm} e_b(a^c)
        + c^p_{pb} e_a(a^b)
    \right] e^a,
\end{align*}
where Eq.~(\ref{EQ-Delta(e^a)}) for $\delta(e^a)$ as been used.

We now compute $II$. We first remark that $II$ rewrites 
\begin{align*}
II&=-\delta a_c de^c\\
&=- \delta a_c(-\frac{1}{2} )c^c_{ab} e^a\wedge e^b\\
&=- \delta \frac{1}{2} a_cc^c_{ab} e^b\wedge e^a,
\end{align*}
which is nothing but $I$ with $e_b(a_a)$ in place of $\frac{1}{2} a_cc^c_{ab}$. Thus, one can recast $II$ under the form $II=(-1)^d\sgn(g)(II_1+II_2)$, 
where the terms $II_1$ et $II_2$ are 
evaluated as follow:
\begin{align*}
II_1&= \frac{1}{2} \sgn(g) (-1)^{d+1}\left[\eta^{aa}e_a(a_c c^c_{ab})e^b-\eta^{cb} e_c(a_{\sss D} c^d_{ab})e^a \right]\\
&= \frac{1}{2} \sgn(g) (-1)^{d+1}\left[\eta^{ad}e_a(a_c c^c_{db})e^b-\eta^{cb} e_c(a_{\sss D} c^d_{ab})e^a \right]\\
&= \frac{1}{2} \sgn(g) (-1)^{d+1}\left[\eta^{ad}e_a(a_c c^c_{db})e^b + \eta^{ad} e_a(a_c c^c_{db})e^b \right]\\
&= \sgn(g) (-1)^{d+1}\left[\eta^{ad}e_a(a_c c^c_{db})e^b\right]\\
  &= \sgn(g) (-1)^{d+1}\left[\eta^{ad}e_a(a_c) c^c_{db}e^b+\eta^{ad}a_ce_a(c^c_{db}) e^b\right],
   \end{align*}
and
\begin{align*}
II_2
&=   \frac{1}{2} \sgn(g) (-1)^{d+1} a_c c^c_{ab} \left[\delta(e^b)e^a+ \eta^{ac}\eta^{bm} \eta_{lf} c^l_{cm} e^f-\delta(e^a) e^b\right]\\
&=\sgn(g) (-1)^{d+1}\left[
 \eta^{bm}a_cc^c_{ba}c^p_{pm}
 + \frac{1}{2}a_c\eta^{mn}\eta^{pq}\eta_{ab} c^c_{mp} c^b_{nq}
 \right] e^a.
\end{align*}

The Maxwell part of $\square$ being determined one focuses on
the second part, namely: $-d\delta(a)$.  Thanks to  Eq.~(\ref{EQ-DeltaOneForm}) one obtains
\begin{align*}
III &=-d\delta(a)\\
&=-d[-e_a(a^a) +a_a\delta(e^a)]\\
 &=e_b[e_a(a^a) + a^ac^p_{pa}]e^b\\
&=e_b\left(e_a(a^a)\right)e^b+ e_b(a^a)c^p_{pa}e^b +a^a e_b(c^p_{pa})e^b.
\end{align*}
Finally
\begin{equation*}
\square a=(-1)^d\sgn(g)( I_1+I_2+II_1+II_2)+ III,
\end{equation*}
which is Eq.~(\ref{EQ-FormGeneralCoeffStruct}).

\subsection{Expression of the term  ${}^n\!B$ of Eq.~(\ref{EQ-coeff-sB})}\label{APP-Calcul-sB}

From Eq.~(\ref{EQ-FormGeneralCoeffStruct}) the terms containing the $n$ index denoted ${}^n\!B$ reads
\begin{equation*}
{}^n\!B = \sum_{m=1}^{9}{}^n\!B^m,
\end{equation*}
where ${}^n\!B^m$ is the $m$-th  term of the r.h.s. of  Eq.~(\ref{EQ-FormGeneralCoeffStruct}):
\begin{equation*}
{}^n\!B^1 =    \eta^{nn} e_n(e_n(a_\mu))e^\mu \ ;\ 
{}^n\!B^2 =(c^n_{\mu \nu}e_n(a^\nu) +c^\nu_{\mu n}e_\nu(a^n) +c^n_{\mu n}e_n(a^n))e^\mu\ \ \mbox{\it etc.}\end{equation*}
Thanks to the relation $a_n= i_{e_n}a = i_{\sss D} a/\sqrt{|y^2|}$, and the 
fact that $e_\mu(y^2)=0$,
one obtains
\begin{align*}
&{}^n\!B^1 =    -\epsilon \frac{1}{|y^2|}(D^2(a_\mu)-D(a_\mu))e^\mu\ ;\ 
{}^n\!B^2 = -\epsilon \frac{1}{|y^2|}   e_\mu(i_{\sss D}a)e^\mu  \ ;\ 
{}^n\!B^3  = -\epsilon \frac{1}{|y^2|} D(a_\mu)e^\mu\ ; \\
&{}^n\!B^4 =       \epsilon \frac{1}{|y^2|} (D(a_\mu)e^\mu-  e_\mu(i_{\sss D}a)e^\mu  ) \ ;\ 
{}^n\!B^5  = - \epsilon \frac{n}{|y^2|} D(a_\mu)e^\mu \ ;\ 
{}^n\!B^6  = \epsilon \frac{1}{|y^2|} a_\mu e^\mu \ ;   \\
&{}^n\!B^7 = - \epsilon \frac{n}{|y^2|} a_\mu e^\mu \ ;\  
{}^n\!B^8 =0   \ ;\ 
{}^n\!B^9 =  \epsilon \frac{1}{|y^2|} a_\mu e^\mu.
\end{align*}
Gathering these terms gives Eq.~($\ref{EQ-coeff-sB}$).

\subsection{Intrinsic form for Eq.~(\ref{EQ-RestricFormRn+2ExtendedFrameVersion})}\label{APP-IntrinsicFormOfEqVect}
We now rewrite the second term of the r.h.s. of Eq.~(\ref{EQ-RestricFormRn+2ExtendedFrameVersion}).
We now show that
\begin{equation}\label{EQ-m*partiala_mu}
    m^*\left((\partial^2_{n+1} a_\mu)e^\mu\right) = m^*\left(H^{-2}\mathcal{L}_{\sss F}^2 (a)\right). 
\end{equation}
To begin with, we have $\mathcal{L}_{\partial_{n+1}}(e^\mu)=0$, since 
\begin{align*}
    \mathcal{L}_{\partial_{n+1}}(e^\mu)&=
    di_{n+1} e^\mu+i_{n+1}de^\mu\\
    &=0-\frac{1}{2}i_{n+1}c^\mu_{\alpha\beta}e^\alpha\wedge\beta\\
    &=0,
\end{align*}
in which we used $c^\mu_{\alpha n+1}=0$. As a consequence, $\mathcal{L}^2_{n+1}(a_\mu)e^\mu=\mathcal{L}^2_{n+1}(a_\mu e^\mu)$ and the result follows using Eq.~(\ref{EQ-FOnCoordn+1}).

Thanks to Eqs.~(\ref{EQ-D=Dp+Dn+1})--(\ref{EQ-FOnCoordn+1}) and 
$y^{n+1} = H^{-1} f$ one has $D_{\sss P} = D - \epsilon H^{-2}f F $, which allows us to eliminate $D_{\sss P}$ from Eq.~(\ref{EQ-RestricFormRn+2ExtendedFrameVersion}). This is achieved as follows.

One first rewrites $\mathcal{L}_{\sss D_{n+1}}$ thanks to the Cartan formula 
under the form: 
\begin{equation*}
{\cal L}_{\phi {\sss V}}=\phi{\cal L}_{\sss V}+d\phi\wedge i_{\sss V},  
\end{equation*}
$\phi$ being a scalar and $V$ a vector. Then 
\begin{equation}\label{EQ-Lie_Dn+1}
    \mathcal{L}_{\sss D_{n+1}} = y^{n+1}\mathcal{L}_{\partial_{n+1}} 
    + dy^{n+1} \wedge i_{\partial_{n+1}}.
\end{equation}
using Eq.~(\ref{EQ-FOnCoordn+1}) and $y^{n+1} = H^{-1} f$ this expression rewrites
$\mathcal{L}_{\sss D_{n+1}} = \epsilon H^{-2}(f{\cal L}_{F}+df\wedge i_F)$,
noting that  $m^*f=1$, $m^*df=0$, one obtains
\begin{equation}\label{EQ-m*L_Dn+1}
m^* \mathcal{L}_{\sss D_{n+1}} =
m^*(\epsilon H^{-2}\mathcal{L}_{\sss F}),
\end{equation}
where, here and after, we note, for any operator $A$, $m^*A$ instead of $m^*\circ A$.
This allows us to recast $m^*\mathcal{L}_{\sss D_{\sss P}}$ under the form
\begin{equation}\label{EQ-m*L_DP}
m^*\mathcal{L}_{\sss D_{\sss P}} = m^*\left(\mathcal{L}_{\sss D} -  \epsilon H^{-2}\mathcal{L}_{\sss F}\right).
\end{equation}

We now consider $m^*\mathcal{L}^2_{\sss D_{\sss P}} = m^*\left(\mathcal{L}_{\sss D}^2 + \mathcal{L}_{\sss D_{n+1}}^2 - \left[\mathcal{L}_{\sss D_{}}, 
\mathcal{L}_{\sss D_{n+1}} \right]_+\right)$, and first compute 
$m^*\mathcal{L}_{\sss D_{n+1}}^2$.
From Eq.~(\ref{EQ-Lie_Dn+1}), taking into account $m^* dy^{n+1} = 0$ and $m^*y^{n+1}=H^{-1}$, one has 
\begin{align*}
 m^*\mathcal{L}_{\sss D_{n+1}}^2 &=  m^*\left(y^{n+1}\mathcal{L}_{\partial_{n+1}} 
+ dy^{n+1} \wedge i_{\partial_{n+1}}\right)  \left(y^{n+1}\mathcal{L}_{\partial_{n+1}} 
+ dy^{n+1} \wedge i_{\partial_{n+1}}\right) \\
 &= m^* y^{n+1}\mathcal{L}_{\partial_{n+1}} y^{n+1}\mathcal{L}_{\partial_{n+1}} 
   + m^* y^{n+1}\mathcal{L}_{\partial_{n+1}} dy^{n+1} \wedge i_{\partial_{n+1}}\\
 &=   m^* y^{n+1}\left(\mathcal{L}_{\partial_{n+1}}(y^{n+1}) \mathcal{L}_{\partial_{n+1}} + y^{n+1} \mathcal{L}^2_{\partial_{n+1}}\right) 
+ m^* y^{n+1}\mathcal{L}_{\partial_{n+1}}(dy^{n+1}) \wedge i_{\partial_{n+1}} \\
 &=m^* \left(y^{n+1}\mathcal{L}_{\partial_{n+1}} + (y^{n+1})^2 \mathcal{L}^2_{\partial_{n+1}}\right). 
\end{align*}
Then we determine the term $m^* \left[\mathcal{L}_{\sss D}, \mathcal{L}_{\sss D_{n+1}} \right]_+$. Noting that $[D, D_{n+1}] = 0$ and thus
$\left[\mathcal{L}_{\sss D}, \mathcal{L}_{\sss D_{n+1}} \right] =0$ one has  
\begin{align*}
  m^* \left[\mathcal{L}_{\sss D}, \mathcal{L}_{\sss D_{n+1}} \right]_+ &= m^* \left(2 \mathcal{L}_{\sss D_{n+1}} \mathcal{L}_{\sss D} \right) \\
  &= m^* \left(2 \left(y^{n+1} \mathcal{L}_{\partial_{n+1}} 
  + dy^{n+1}\wedge i_{\partial_{n+1}}\right) \mathcal{L}_{\sss D} \right)\\
  &= m^* 2 \left(y^{n+1} \mathcal{L}_{\partial_{n+1}} \mathcal{L}_{\sss D} \right),
\end{align*}
where $m^*dy^{n+1} = 0$ has been used. 

Gathering the above results and using Eq.~(\ref{EQ-FOnCoordn+1}) and $y^{n+1} = H^{-1} f$ gives
\begin{equation}\label{EQ-m*L^2_DP}
m^*\mathcal{L}^2_{\sss D_{\sss P}} =    m^* \left(\mathcal{L}^2_{\sss D} + H^{-4} \mathcal{L}^2_{\sss F} 
+ \epsilon H^{-4} \mathcal{L}^2_{\sss F}(1 - 2 \mathcal{L}_{\sss D})\right).
\end{equation}

The last term of Eq.~(\ref{EQ-RestricFormRn+2ExtendedFrameVersion}) containing $D_{\sss P}$, rewrites as  
\begin{equation}\label{EQ-m*di_DP}
m^*di_{\sss D_{\sss P}} = m^*(di_{\sss D} - \epsilon H^{-2} di_{\sss F}).
\end{equation}

Finally, using the above relations Eqs.~(\ref{EQ-m*partiala_mu}), (\ref{EQ-m*L_DP}), (\ref{EQ-m*L^2_DP}) and (\ref{EQ-m*di_DP}),  
the second term of the r.h.s. of Eq.~(\ref{EQ-RestricFormRn+2ExtendedFrameVersion}):
\begin{equation*}
    X:= m^*\left[\epsilon(\partial^2_{n+1} a_\mu)e^\mu - \epsilon H^2 
\left({\cal L}^2_{\sss D_{\sss P}}(a)+(n-3){\cal L}_{\sss D_{\sss P}}(a)+ 2\,d(i_{\sss D_{\sss P}}a)\right)\right],
\end{equation*}
writes
\begin{equation*}
X = 
 m^*\left[ -\epsilon H^2\left( 
\mathcal{L}_{\sss D}^2 + (n-3)\mathcal{L}_{\sss D} + 2 d i_{\sss D}\right) 
+ 2\mathcal{L}_{\sss F}\mathcal{L}_{\sss D}
+ (n - 4) \mathcal{L}_{\sss F} + 2 d i_{\sss F}\right](a).
\end{equation*}

\subsection{Derivation of Eq.~(\ref{EQ-RestricScalarRn+2IntrinsicVersion})}\label{APP-IntrinsicFormOfEqScal}
The Eq.~(\ref{EQ-ScalarGeneralCoeffStruct}), as for the one-form, splits into three terms:
\begin{align*}
{}^r\!C&=\eta^{\mu\nu}[e_\mu e_\nu(\phi)+e_\mu(\phi) c^\alpha_{\alpha\mu}],\\
{}^n\!C&=\eta^{nn}[e_n e_n(\phi)+e_n(\phi) c^\alpha_{\alpha n}],\\
{}^{n+1}\!C&=\eta^{n+1 n+1}[\partial_{n+1}\partial_{n+1}(\phi)+\partial_{n+1}(\phi) c^\alpha_{\alpha n+1}].
\end{align*} 

As already said $m^*{}^r\!C=\square_{\sss \Sigma} \phi_{\sss \Sigma}$. 
Now, using the notations of Sec.~\ref{SUBSEC-AdaptedFramesRn+2},  and $y^2_{\sss P} := y^{\sss A}y_{\sss A}$, the two remaining terms read
\begin{align*}
{}^n\!C&=\eta^{nn}[e_n e_n(\phi)+e_n(\phi) c^\alpha_{\alpha n}]\\
&=-\epsilon\left[\frac{1}{\sqrt{|y^2_{\sss P} |}} (D_{\sss P}) \frac{1}{\sqrt{|y^2_{\sss P}|}}(D_{\sss P})\phi+\frac{n}{|y^2_{\sss P}|}(D_{\sss P})\phi\right]\\
&=-\frac{\epsilon}{|y^2_{\sss P} |}[ (D_{\sss P})^2 \phi- (D_{\sss P})\phi +n(D_{\sss P})\phi]\\
&=-\frac{\epsilon}{|y^2_{\sss P} |}[ D^2 \phi-2D_{n+1}D\phi+(y^{n+1})^2\partial_{n+1}^2\phi+D_{n+1}\phi +(n-1)(D - D_{n+1})\phi],
\end{align*}
and ${}^{n+1}\!C = \epsilon\partial_{n+1}^2\phi$.

Using $y^{n+1} = H^{-1}f$ one remarks that the term ${}^{n+1}\!C$ simplify with a term in ${}^n\!C$:
\begin{equation*}
m^*\left(-\frac{\epsilon}{|y^2_{\sss P}|} (y^{n+1})^2\partial_{n+1}^2\phi + \epsilon\partial_{n+1}^2\phi\right) = 
-\frac{\epsilon}{H^{-2}} H^{-2}\partial_{n+1}^2\phi + \epsilon\partial_{n+1}^2\phi = 0,
\end{equation*}
and one thus has successively
\begin{align*}
m^*(\square_{n+2}\phi)&=
\square_{\sss \Sigma}\phi_{\sss \Sigma} -\epsilon H^2\left[D^2\phi-2D_{n+1}D\phi+D_{n+1}\phi+(n-1)
(D\phi-D_{n+1})\phi\right]_{\sss \Sigma}\\
&=\square_{\sss \Sigma}\phi_{\sss \Sigma} -\epsilon H^2\left[D^2\phi-2D_{n+1}D\phi+D_{n+1}\phi+(n-1)(D_{\sss P})\phi\right]_{\sss \Sigma}\\
&=\square_{\sss \Sigma}\phi_{\sss \Sigma} -\epsilon H^2\left[D^2\phi +(n-1)D\phi-2H^{-1}\partial_{n+1}D\phi+(2-n)H^{-1}\partial_{n+1}\phi\right]_{\sss \Sigma}\\
&=\square_{\sss \Sigma}\phi_{\sss \Sigma} -\epsilon H^2\left[D^2\phi +(n-1)D\phi-2\epsilon H^{-2}FD\phi+\epsilon(2-n)H^{-2}F\phi\right]_{\sss \Sigma}\\
&=\square_{\sss \Sigma}\phi_{\sss \Sigma} -\epsilon H^2\left[D^2\phi +(n-1)D\phi\right]_{\sss \Sigma} +\left[2FD\phi+(n-2)F\phi\right]_{\sss \Sigma},
\end{align*}
which is Eq.~(\ref{EQ-RestricScalarRn+2IntrinsicVersion}).

%%%%
%\bibliography{BoxVect-bibliography}

\begin{thebibliography}{29}%
    \makeatletter
    \providecommand \@ifxundefined [1]{%
        \@ifx{#1\undefined}
    }%
    \providecommand \@ifnum [1]{%
        \ifnum #1\expandafter \@firstoftwo
        \else \expandafter \@secondoftwo
        \fi
    }%
    \providecommand \@ifx [1]{%
        \ifx #1\expandafter \@firstoftwo
        \else \expandafter \@secondoftwo
        \fi
    }%
    \providecommand \natexlab [1]{#1}%
    \providecommand \enquote  [1]{``#1''}%
    \providecommand \bibnamefont  [1]{#1}%
    \providecommand \bibfnamefont [1]{#1}%
    \providecommand \citenamefont [1]{#1}%
    \providecommand \href@noop [0]{\@secondoftwo}%
    \providecommand \href [0]{\begingroup \@sanitize@url \@href}%
    \providecommand \@href[1]{\@@startlink{#1}\@@href}%
    \providecommand \@@href[1]{\endgroup#1\@@endlink}%
    \providecommand \@sanitize@url [0]{\catcode `\\12\catcode `\$12\catcode
        `\&12\catcode `\#12\catcode `\^12\catcode `\_12\catcode `\%12\relax}%
    \providecommand \@@startlink[1]{}%
    \providecommand \@@endlink[0]{}%
    \providecommand \url  [0]{\begingroup\@sanitize@url \@url }%
    \providecommand \@url [1]{\endgroup\@href {#1}{\urlprefix }}%
    \providecommand \urlprefix  [0]{URL }%
    \providecommand \Eprint [0]{\href }%
    \providecommand \doibase [0]{https://doi.org/}%
    \providecommand \selectlanguage [0]{\@gobble}%
    \providecommand \bibinfo  [0]{\@secondoftwo}%
    \providecommand \bibfield  [0]{\@secondoftwo}%
    \providecommand \translation [1]{[#1]}%
    \providecommand \BibitemOpen [0]{}%
    \providecommand \bibitemStop [0]{}%
    \providecommand \bibitemNoStop [0]{.\EOS\space}%
    \providecommand \EOS [0]{\spacefactor3000\relax}%
    \providecommand \BibitemShut  [1]{\csname bibitem#1\endcsname}%
    \let\auto@bib@innerbib\@empty
    %</preamble>
    \bibitem [{\citenamefont {Pavsic}\ and\ \citenamefont
        {Tapia}(2000)}]{Pavsic:2000qy}%
    \BibitemOpen
    \bibfield  {author} {\bibinfo {author} {\bibfnamefont {M.}~\bibnamefont
            {Pavsic}}\ and\ \bibinfo {author} {\bibfnamefont {V.}~\bibnamefont 
            {Tapia}},\
    }\bibfield  {title} {\bibinfo {title} {{Resource letter on geometrical
                results for embeddings and branes}},\ }\Eprint
    {https://arxiv.org/abs/gr-qc/0010045} {arXiv:gr-qc/0010045}  (\bibinfo 
    {year}
    {2000})\BibitemShut {NoStop}%
    \bibitem [{\citenamefont {Paston}\ and\ \citenamefont
        {Sheykin}(2013)}]{Paston:2013uia}%
    \BibitemOpen
    \bibfield  {author} {\bibinfo {author} {\bibfnamefont {S.~A.}\ \bibnamefont
            {Paston}}\ and\ \bibinfo {author} {\bibfnamefont {A.~A.}\ 
            \bibnamefont
            {Sheykin}},\ }\bibfield  {title} {\bibinfo {title} {{Embeddings for 
            solutions
                of Einstein equations}},\ }\href 
                {https://doi.org/10.1007/s11232-013-0067-4}
    {\bibfield  {journal} {\bibinfo  {journal} {Theor. Math. Phys.}\ }\textbf
        {\bibinfo {volume} {175}},\ \bibinfo {pages} {806} (\bibinfo {year}
        {2013})},\ \Eprint {https://arxiv.org/abs/1306.4826} {arXiv:1306.4826
        [gr-qc]} \BibitemShut {NoStop}%
    \bibitem [{\citenamefont {Akbar}(2017)}]{Akbar:2017vja}%
    \BibitemOpen
    \bibfield  {author} {\bibinfo {author} {\bibfnamefont {M.~M.}\ \bibnamefont
            {Akbar}},\ }\bibfield  {title} {\bibinfo {title} {{Embedding FLRW 
            geometries
                in pseudo-Euclidean and anti\textendash{}de Sitter spaces}},\ 
                }\href
    {https://doi.org/10.1103/PhysRevD.95.064058} {\bibfield  {journal} {\bibinfo
            {journal} {Phys. Rev. D}\ }\textbf {\bibinfo {volume} {95}},\ 
            \bibinfo
        {pages} {064058} (\bibinfo {year} {2017})},\ \Eprint
    {https://arxiv.org/abs/1702.00987} {arXiv:1702.00987 [gr-qc]} \BibitemShut
    {NoStop}%
    \bibitem [{\citenamefont {Dunajski}\ and\ \citenamefont
        {Tod}(2019)}]{Dunajski:2018xoa}%
    \BibitemOpen
    \bibfield  {author} {\bibinfo {author} {\bibfnamefont {M.}~\bibnamefont
            {Dunajski}}\ and\ \bibinfo {author} {\bibfnamefont 
            {P.}~\bibnamefont {Tod}},\
    }\bibfield  {title} {\bibinfo {title} {{Conformally isometric embeddings and
                Hawking temperature}},\ }\href 
                {https://doi.org/10.1088/1361-6382/ab2068}
    {\bibfield  {journal} {\bibinfo  {journal} {Class. Quant. Grav.}\ }\textbf
        {\bibinfo {volume} {36}},\ \bibinfo {pages} {125005} (\bibinfo {year}
        {2019})},\ \Eprint {https://arxiv.org/abs/1812.05468} {arXiv:1812.05468
        [gr-qc]} \BibitemShut {NoStop}%
    \bibitem [{\citenamefont {Hong}\ \emph {et~al.}(2022)\citenamefont {Hong},
        \citenamefont {Kim},\ and\ \citenamefont {Park}}]{Hong:2020dow}%
    \BibitemOpen
    \bibfield  {author} {\bibinfo {author} {\bibfnamefont {S.-T.}\ \bibnamefont
            {Hong}}, \bibinfo {author} {\bibfnamefont {Y.-W.}\ \bibnamefont 
            {Kim}},\ and\
        \bibinfo {author} {\bibfnamefont {Y.-J.}\ \bibnamefont {Park}},\ 
        }\bibfield
    {title} {\bibinfo {title} {{GEMS embeddings of Schwarzschild and RN black
                holes in Painlev\'e-Gullstrand spacetimes}},\ }\href
    {https://doi.org/10.3390/universe8010015} {\bibfield  {journal} {\bibinfo
            {journal} {Universe}\ }\textbf {\bibinfo {volume} {8}},\ \bibinfo 
            {pages}
        {15} (\bibinfo {year} {2022})},\ \Eprint 
        {https://arxiv.org/abs/2011.08351}
    {arXiv:2011.08351 [gr-qc]} \BibitemShut {NoStop}%
    \bibitem [{\citenamefont {Dirac}(1935)}]{Dirac:1935zz}%
    \BibitemOpen
    \bibfield  {author} {\bibinfo {author} {\bibfnamefont {P.~A.~M.}\
            \bibnamefont {Dirac}},\ }\bibfield  {title} {\bibinfo {title} {{The 
            Electron
                Wave Equation in De-Sitter Space}},\ }\href 
                {https://doi.org/10.2307/1968649}
    {\bibfield  {journal} {\bibinfo  {journal} {Annals Math.}\ }\textbf 
    {\bibinfo
            {volume} {36}},\ \bibinfo {pages} {657} (\bibinfo {year} 
            {1935})}\BibitemShut
    {NoStop}%
    \bibitem [{\citenamefont {Dirac}(1936)}]{Dirac:1936fq}%
    \BibitemOpen
    \bibfield  {author} {\bibinfo {author} {\bibfnamefont {P.~A.~M.}\
            \bibnamefont {Dirac}},\ }\bibfield  {title} {\bibinfo {title} {{Wave
                equations in conformal space}},\ }\href 
                {https://doi.org/10.2307/1968455}
    {\bibfield  {journal} {\bibinfo  {journal} {Annals Math.}\ }\textbf 
    {\bibinfo
            {volume} {37}},\ \bibinfo {pages} {429} (\bibinfo {year} 
            {1936})}\BibitemShut
    {NoStop}%
    %%CITATION = ANMAA,37,429;%%
    \bibitem [{\citenamefont {Fronsdal}(1965)}]{Fronsdal:1965zzb}%
    \BibitemOpen
    \bibfield  {author} {\bibinfo {author} {\bibfnamefont {C.}~\bibnamefont
            {Fronsdal}},\ }\bibfield  {title} {\bibinfo {title} {{Elementary 
            Particles in
                a Curved Space}},\ }\href 
                {https://doi.org/10.1103/RevModPhys.37.221}
    {\bibfield  {journal} {\bibinfo  {journal} {Rev. Mod. Phys.}\ }\textbf
        {\bibinfo {volume} {37}},\ \bibinfo {pages} {221} (\bibinfo {year}
        {1965})}\BibitemShut {NoStop}%
    \bibitem [{\citenamefont {Fronsdal}(1974)}]{Fronsdal:1974ew}%
    \BibitemOpen
    \bibfield  {author} {\bibinfo {author} {\bibfnamefont {C.}~\bibnamefont
            {Fronsdal}},\ }\bibfield  {title} {\bibinfo {title} {{Elementary 
            particles in
                a curved space. II}},\ }\href 
                {https://doi.org/10.1103/PhysRevD.10.589}
    {\bibfield  {journal} {\bibinfo  {journal} {Phys. Rev. D}\ }\textbf 
    {\bibinfo
            {volume} {10}},\ \bibinfo {pages} {589} (\bibinfo {year} 
            {1974})}\BibitemShut
    {NoStop}%
    \bibitem [{\citenamefont {Fronsdal}\ and\ \citenamefont
        {Haugen}(1975)}]{Fronsdal:1975eq}%
    \BibitemOpen
    \bibfield  {author} {\bibinfo {author} {\bibfnamefont {C.}~\bibnamefont
            {Fronsdal}}\ and\ \bibinfo {author} {\bibfnamefont {R.~B.}\ 
            \bibnamefont
            {Haugen}},\ }\bibfield  {title} {\bibinfo {title} {{Elementary 
            Particles in a
                Curved Space. III.}},\ }\href 
                {https://doi.org/10.1103/PhysRevD.12.3810}
    {\bibfield  {journal} {\bibinfo  {journal} {Phys. Rev. D}\ }\textbf 
    {\bibinfo
            {volume} {12}},\ \bibinfo {pages} {3810} (\bibinfo {year}
        {1975})}\BibitemShut {NoStop}%
    \bibitem [{\citenamefont {Fronsdal}(1975)}]{PhysRevD.12.3819}%
    \BibitemOpen
    \bibfield  {author} {\bibinfo {author} {\bibfnamefont {C.}~\bibnamefont
            {Fronsdal}},\ }\bibfield  {title} {\bibinfo {title} {{Elementary 
            particles in
                a curved space. IV.}},\ }\href 
                {https://doi.org/10.1103/PhysRevD.12.3819}
    {\bibfield  {journal} {\bibinfo  {journal} {Phys. Rev. D}\ }\textbf 
    {\bibinfo
            {volume} {12}},\ \bibinfo {pages} {3819} (\bibinfo {year}
        {1975})}\BibitemShut {NoStop}%
    \bibitem [{\citenamefont {Fronsdal}(1979)}]{Fronsdal:1978vb}%
    \BibitemOpen
    \bibfield  {author} {\bibinfo {author} {\bibfnamefont {C.}~\bibnamefont
            {Fronsdal}},\ }\bibfield  {title} {\bibinfo {title} {{Singletons and
                Massless, Integral Spin Fields on de Sitter Space (Elementary 
                Particles in a
                Curved Space. 7.}},\ }\href 
                {https://doi.org/10.1103/PhysRevD.20.848}
    {\bibfield  {journal} {\bibinfo  {journal} {Phys. Rev. D}\ }\textbf 
    {\bibinfo
            {volume} {20}},\ \bibinfo {pages} {848} (\bibinfo {year} 
            {1979})}\BibitemShut
    {NoStop}%
    \bibitem [{\citenamefont {Fang}\ and\ \citenamefont
        {Fronsdal}(1980)}]{Fang:1979hq}%
    \BibitemOpen
    \bibfield  {author} {\bibinfo {author} {\bibfnamefont {J.}~\bibnamefont
            {Fang}}\ and\ \bibinfo {author} {\bibfnamefont {C.}~\bibnamefont
            {Fronsdal}},\ }\bibfield  {title} {\bibinfo {title} {{Massless, 
            Half Integer
                Spin Fields in De Sitter Space}},\ }\href
    {https://doi.org/10.1103/PhysRevD.22.1361} {\bibfield  {journal} {\bibinfo
            {journal} {Phys. Rev. D}\ }\textbf {\bibinfo {volume} {22}},\ 
            \bibinfo
        {pages} {1361} (\bibinfo {year} {1980})}\BibitemShut {NoStop}%
    \bibitem [{\citenamefont {Gazeau}\ and\ \citenamefont
        {Hans}(1988)}]{Gazeau:1987nu}%
    \BibitemOpen
    \bibfield  {author} {\bibinfo {author} {\bibfnamefont {J.~P.}\ \bibnamefont
            {Gazeau}}\ and\ \bibinfo {author} {\bibfnamefont {M.}~\bibnamefont 
            {Hans}},\
    }\bibfield  {title} {\bibinfo {title} {{Integral Spin Fields on (3+2) de
                Sitter space}},\ }\href {https://doi.org/10.1063/1.528094} 
                {\bibfield
        {journal} {\bibinfo  {journal} {J. Math. Phys.}\ }\textbf {\bibinfo 
        {volume}
            {29}},\ \bibinfo {pages} {2533} (\bibinfo {year} 
            {1988})}\BibitemShut
    {NoStop}%
    \bibitem [{\citenamefont {Gazeau}\ and\ \citenamefont
        {Takook}(2000)}]{Gazeau:1999xn}%
    \BibitemOpen
    \bibfield  {author} {\bibinfo {author} {\bibfnamefont {J.-P.}\ \bibnamefont
            {Gazeau}}\ and\ \bibinfo {author} {\bibfnamefont {M.~V.}\ 
            \bibnamefont
            {Takook}},\ }\bibfield  {title} {\bibinfo {title} {{'Massive' 
            vector field in
                de Sitter space}},\ }\href {https://doi.org/10.1063/1.1287641} 
                {\bibfield
        {journal} {\bibinfo  {journal} {J. Math. Phys.}\ }\textbf {\bibinfo 
        {volume}
            {41}},\ \bibinfo {pages} {5920} (\bibinfo {year} {2000})},\ \Eprint
    {https://arxiv.org/abs/gr-qc/9912080} {arXiv:gr-qc/9912080} \BibitemShut
    {NoStop}%
    \bibitem [{\citenamefont {Garidi}\ \emph {et~al.}(2003)\citenamefont 
    {Garidi},
        \citenamefont {Gazeau},\ and\ \citenamefont {Takook}}]{Garidi:2003bg}%
    \BibitemOpen
    \bibfield  {author} {\bibinfo {author} {\bibfnamefont {T.}~\bibnamefont
            {Garidi}}, \bibinfo {author} {\bibfnamefont {J.~P.}\ \bibnamefont 
            {Gazeau}},\
        and\ \bibinfo {author} {\bibfnamefont {M.~V.}\ \bibnamefont {Takook}},\
    }\bibfield  {title} {\bibinfo {title} {{'Massive' spin two field in de 
    Sitter
                space}},\ }\href {https://doi.org/10.1063/1.1599055} 
                {\bibfield  {journal}
        {\bibinfo  {journal} {J. Math. Phys.}\ }\textbf {\bibinfo {volume} 
        {44}},\
        \bibinfo {pages} {3838} (\bibinfo {year} {2003})},\ \Eprint
    {https://arxiv.org/abs/hep-th/0302022} {arXiv:hep-th/0302022} \BibitemShut
    {NoStop}%
    \bibitem [{\citenamefont {Huguet}\ \emph {et~al.}(2006)\citenamefont 
    {Huguet},
        \citenamefont {Queva},\ and\ \citenamefont {Renaud}}]{Huguet:2006fe}%
    \BibitemOpen
    \bibfield  {author} {\bibinfo {author} {\bibfnamefont {E.}~\bibnamefont
            {Huguet}}, \bibinfo {author} {\bibfnamefont {J.}~\bibnamefont 
            {Queva}},\ and\
        \bibinfo {author} {\bibfnamefont {J.}~\bibnamefont {Renaud}},\ 
        }\bibfield
    {title} {\bibinfo {title} {{Conformally related massless fields in dS, AdS
                and Minkowski spaces}},\ }\href 
                {https://doi.org/10.1103/PhysRevD.73.084025}
    {\bibfield  {journal} {\bibinfo  {journal} {Phys. Rev.}\ }\textbf {\bibinfo
            {volume} {D73}},\ \bibinfo {pages} {084025} (\bibinfo {year} 
            {2006})},\
    \Eprint {https://arxiv.org/abs/gr-qc/0603031} {arXiv:gr-qc/0603031 [gr-qc]}
    \BibitemShut {NoStop}%
    %%CITATION = GR-QC/0603031;%%
    \bibitem [{\citenamefont {Queva}(2009)}]{queva:tel-00503186}%
    \BibitemOpen
    \bibfield  {author} {\bibinfo {author} {\bibfnamefont {J.}~\bibnamefont
            {Queva}},\ }\emph {\bibinfo {title} {{Sur quelques probl{\`e}mes de
                quantification : en espace-temps de de Sitter et par {\'e}tats
                coh{\'e}rents}}},\ \href 
                {https://tel.archives-ouvertes.fr/tel-00503186}
    {\bibinfo {type} {Theses}},\ \bibinfo  {school} {{Universit{\'e}
            Paris-Diderot - Paris VII}} (\bibinfo {year} {2009})\BibitemShut 
            {NoStop}%
    \bibitem [{\citenamefont {Faci}\ \emph {et~al.}(2009)\citenamefont {Faci},
        \citenamefont {Huguet}, \citenamefont {Queva},\ and\ \citenamefont
        {Renaud}}]{Faci:2009un}%
    \BibitemOpen
    \bibfield  {author} {\bibinfo {author} {\bibfnamefont {S.}~\bibnamefont
            {Faci}}, \bibinfo {author} {\bibfnamefont {E.}~\bibnamefont 
            {Huguet}},
        \bibinfo {author} {\bibfnamefont {J.}~\bibnamefont {Queva}},\ and\ 
        \bibinfo
        {author} {\bibfnamefont {J.}~\bibnamefont {Renaud}},\ }\bibfield  
        {title}
    {\bibinfo {title} {{Conformally covariant quantization of Maxwell field in 
    de
                Sitter space}},\ }\href 
                {https://doi.org/10.1103/PhysRevD.80.124005}
    {\bibfield  {journal} {\bibinfo  {journal} {Phys. Rev. D}\ }\textbf 
    {\bibinfo
            {volume} {80}},\ \bibinfo {pages} {124005} (\bibinfo {year} 
            {2009})},\
    \Eprint {https://arxiv.org/abs/0910.1279} {arXiv:0910.1279 [gr-qc]}
    \BibitemShut {NoStop}%
    \bibitem [{\citenamefont {Huguet}\ and\ \citenamefont
        {Renaud}(2013)}]{Huguet:2013tv}%
    \BibitemOpen
    \bibfield  {author} {\bibinfo {author} {\bibfnamefont {E.}~\bibnamefont
            {Huguet}}\ and\ \bibinfo {author} {\bibfnamefont {J.}~\bibnamefont
            {Renaud}},\ }\bibfield  {title} {\bibinfo {title} {{Conformally 
            invariant
                formalism for the electromagnetic field with currents in 
                Robertson-Walker
                spaces}},\ }\href {https://doi.org/10.1063/1.4791688} 
                {\bibfield  {journal}
        {\bibinfo  {journal} {J. Math. Phys.}\ }\textbf {\bibinfo {volume} 
        {54}},\
        \bibinfo {pages} {022304} (\bibinfo {year} {2013})},\ \Eprint
    {https://arxiv.org/abs/1301.7646} {arXiv:1301.7646 [gr-qc]} \BibitemShut
    {NoStop}%
    %%CITATION = ARXIV:1301.7646;%%
    \bibitem [{\citenamefont {Takook}(2014)}]{Takook:2014paa}%
    \BibitemOpen
    \bibfield  {author} {\bibinfo {author} {\bibfnamefont {M.~V.}\ \bibnamefont
            {Takook}},\ }\bibfield  {title} {\bibinfo {title} {{Quantum Field 
            Theory in
                de Sitter Universe: Ambient Space Formalism}},\ }\Eprint
    {https://arxiv.org/abs/1403.1204} {arXiv:1403.1204 [gr-qc]}  (\bibinfo 
    {year}
    {2014})\BibitemShut {NoStop}%
    \bibitem [{\citenamefont {Zapata}\ \emph {et~al.}(2017)\citenamefont 
    {Zapata},
        \citenamefont {Belokogne}, \citenamefont {Huguet}, \citenamefont 
        {Queva},\
        and\ \citenamefont {Renaud}}]{Zapata:2017gqg}%
    \BibitemOpen
    \bibfield  {author} {\bibinfo {author} {\bibfnamefont {J.~P.~A.}\
            \bibnamefont {Zapata}}, \bibinfo {author} {\bibfnamefont 
            {A.}~\bibnamefont
            {Belokogne}}, \bibinfo {author} {\bibfnamefont {E.}~\bibnamefont 
            {Huguet}},
        \bibinfo {author} {\bibfnamefont {J.}~\bibnamefont {Queva}},\ and\ 
        \bibinfo
        {author} {\bibfnamefont {J.}~\bibnamefont {Renaud}},\ }\bibfield  
        {title}
    {\bibinfo {title} {{Friedmann-Lema\^\i{}tre-Robertson-Walker spaces as
                submanifolds of $\mathbb{R}^6$: Restriction to the Klein-Gordon 
                operator}},\
    }\href {https://doi.org/10.1063/1.4998179} {\bibfield  {journal} {\bibinfo
            {journal} {J. Math. Phys.}\ }\textbf {\bibinfo {volume} {58}},\ 
            \bibinfo
        {pages} {113503} (\bibinfo {year} {2017})},\ \Eprint
    {https://arxiv.org/abs/1711.10771} {arXiv:1711.10771 [math-ph]} \BibitemShut
    {NoStop}%
    \bibitem [{\citenamefont {Pethybridge}\ and\ \citenamefont
        {Schaub}(2021)}]{Pethybridge:2021rwf}%
    \BibitemOpen
    \bibfield  {author} {\bibinfo {author} {\bibfnamefont {B.}~\bibnamefont
            {Pethybridge}}\ and\ \bibinfo {author} {\bibfnamefont 
            {V.}~\bibnamefont
            {Schaub}},\ }\bibfield  {title} {\bibinfo {title} {{Tensors and 
            Spinors in de
                Sitter Space}},\ }\Eprint {https://arxiv.org/abs/2111.14899}
    {arXiv:2111.14899 [hep-th]}  (\bibinfo {year} {2021})\BibitemShut {NoStop}%
    \bibitem [{\citenamefont {Jensen}\ and\ \citenamefont
        {Koppe}(1971)}]{JENSEN1971586}%
    \BibitemOpen
    \bibfield  {author} {\bibinfo {author} {\bibfnamefont {H.}~\bibnamefont
            {Jensen}}\ and\ \bibinfo {author} {\bibfnamefont {H.}~\bibnamefont 
            {Koppe}},\
    }\bibfield  {title} {\bibinfo {title} {Quantum mechanics with constraints},\
    }\href {https://doi.org/https://doi.org/10.1016/0003-4916(71)90031-5}
    {\bibfield  {journal} {\bibinfo  {journal} {Annals of Physics}\ }\textbf
        {\bibinfo {volume} {63}},\ \bibinfo {pages} {586} (\bibinfo {year}
        {1971})}\BibitemShut {NoStop}%
    \bibitem [{\citenamefont {da~Costa}(1986)}]{Costa_1986}%
    \BibitemOpen
    \bibfield  {author} {\bibinfo {author} {\bibfnamefont {R.~C.~T.}\
            \bibnamefont {da~Costa}},\ }\bibfield  {title} {\bibinfo {title} 
            {Constraints
            in quantum mechanics},\ }\href 
            {https://doi.org/10.1088/0143-0807/7/4/010}
    {\bibfield  {journal} {\bibinfo  {journal} {European Journal of Physics}\
        }\textbf {\bibinfo {volume} {7}},\ \bibinfo {pages} {269} (\bibinfo 
        {year}
        {1986})}\BibitemShut {NoStop}%
    \bibitem [{\citenamefont {Ikegami}\ \emph {et~al.}(1992)\citenamefont
        {Ikegami}, \citenamefont {Nagaoka}, \citenamefont {Takagi},\ and\
        \citenamefont {Tanzawa}}]{10.1143/ptp/88.2.229}%
    \BibitemOpen
    \bibfield  {author} {\bibinfo {author} {\bibfnamefont {M.}~\bibnamefont
            {Ikegami}}, \bibinfo {author} {\bibfnamefont {Y.}~\bibnamefont 
            {Nagaoka}},
        \bibinfo {author} {\bibfnamefont {S.}~\bibnamefont {Takagi}},\ and\ 
        \bibinfo
        {author} {\bibfnamefont {T.}~\bibnamefont {Tanzawa}},\ }\bibfield  
        {title}
    {\bibinfo {title} {{Quantum Mechanics of a Particle on a Curved Surface:
                Comparison of Three Different Approaches}},\ }\href
    {https://doi.org/10.1143/ptp/88.2.229} {\bibfield  {journal} {\bibinfo
            {journal} {Progress of Theoretical Physics}\ }\textbf {\bibinfo 
            {volume}
            {88}},\ \bibinfo {pages} {229} (\bibinfo {year} {1992})}\BibitemShut
    {NoStop}%
    \bibitem [{\citenamefont {Matsutani}(1993)}]{PhysRevA.47.686}%
    \BibitemOpen
    \bibfield  {author} {\bibinfo {author} {\bibfnamefont {S.}~\bibnamefont
            {Matsutani}},\ }\bibfield  {title} {\bibinfo {title} {Quantum field 
            theory on
            curved low-dimensional space embedded in three-dimensional space},\ 
            }\href
    {https://doi.org/10.1103/PhysRevA.47.686} {\bibfield  {journal} {\bibinfo
            {journal} {Phys. Rev. A}\ }\textbf {\bibinfo {volume} {47}},\ 
            \bibinfo
        {pages} {686} (\bibinfo {year} {1993})}\BibitemShut {NoStop}%
    \bibitem [{\citenamefont {Huguet}\ \emph {et~al.}(2020)\citenamefont 
    {Huguet},
        \citenamefont {Queva},\ and\ \citenamefont {Renaud}}]{Huguet:2016szt}%
    \BibitemOpen
    \bibfield  {author} {\bibinfo {author} {\bibfnamefont {E.}~\bibnamefont
            {Huguet}}, \bibinfo {author} {\bibfnamefont {J.}~\bibnamefont 
            {Queva}},\ and\
        \bibinfo {author} {\bibfnamefont {J.}~\bibnamefont {Renaud}},\ 
        }\bibfield
    {title} {\bibinfo {title} {{Massive scalar field on (A)dS space from a
                massless conformal field in $\mathbb{R}^6$}},\ }\href
    {https://doi.org/10.1063/1.5132893} {\bibfield  {journal} {\bibinfo
            {journal} {J. Math. Phys.}\ }\textbf {\bibinfo {volume} {61}},\ 
            \bibinfo
        {pages} {053506} (\bibinfo {year} {2020})},\ \Eprint
    {https://arxiv.org/abs/1606.07611} {arXiv:1606.07611 [gr-qc]} \BibitemShut
    {NoStop}%
    \bibitem [{\citenamefont {Fecko}(2006)}]{Fecko:2006}%
    \BibitemOpen
    \bibfield  {author} {\bibinfo {author} {\bibfnamefont {M.}~\bibnamefont
            {Fecko}},\ }\href@noop {} {\emph {\bibinfo {title} {{Differential 
            Geometry
                    and Lie Groups for Physicists}}}}\ (\bibinfo  {publisher} 
                    {Cambridge
        University Press},\ \bibinfo {year} {2006})\BibitemShut {NoStop}%
\end{thebibliography}
%
\end{document}